\documentclass[journal, 10pt, romanappendices]{IEEEtran}
\usepackage[draft, blank]{ieeefig}
\usepackage{amsmath,graphicx}

\usepackage{cite}
\usepackage{amsopn}
\usepackage{float}
\restylefloat{table}
\usepackage{amsthm}

\usepackage{latexsym}
\usepackage{amsfonts}
\usepackage{amssymb}
\usepackage{mathrsfs}
\usepackage{bm}
\usepackage{yhmath}
\usepackage{epstopdf}

\newtheorem{theorem}{Theorem}[section]
\newtheorem{corollary}{Corollary}[section]
\newtheorem{lemma}{Lemma}[section]

\theoremstyle{definition}
\newtheorem{definition}{Definition}[section]

\newtheorem*{remark}{Remark}

\theoremstyle{condition}
\newtheorem{condition}{Condition}[section]

\DeclareMathAlphabet{\bi}{OML}{cmm}{b}{it}
\DeclareMathAlphabet{\bcal}{OMS}{cmsy}{b}{n}
\DeclareMathAlphabet{\brmn}{OT1}{cmr}{bx}{n}

\newcommand{\bfeta}{\boldsymbol{\eta}}
\newcommand{\bpsi}{\boldsymbol{\varphi}}

\def\x{\mathbf{x}}
\def\X{\mathbf{X}}
\def\Z{\mathbf{Z}}
\def\E{\mathbf{E}}
\def\V{\mathbf{V}}

\def \a{\mathbf{a}}
\def \v{\mathbf{v}}
\def \y{\mathbf{y}}
\def \z{\mathbf{z}}

\def \e{\mathbf{e}}

\title{A Deterministic Theory for Exact Non-Convex Phase Retrieval} 

\author{\IEEEauthorblockN{Bariscan Yonel\IEEEauthorrefmark{1}\IEEEauthorrefmark{2}, \IEEEmembership{Student Member, IEEE}
    and Birsen Yazici\IEEEauthorrefmark{1}\IEEEauthorrefmark{3}, \IEEEmembership{Senior Member, IEEE}}
\thanks{\IEEEauthorrefmark{1}Yonel and Yazici are with the Department of Electrical, Computer and Systems Engineering,
Rensselaer Polytechnic Institute, 110 8th Street, Troy, NY 12180 USA, E-mail:
\IEEEauthorrefmark{2}yonelb@rpi.edu, \IEEEauthorrefmark{3}yazici@ecse.rpi.edu,
Phone: (518)-276 2905, Fax: (518)-276 6261.}
\thanks{This work was supported by the Air Force Office of Scientific Research
(AFOSR) under the agreement FA9550-19-1-0284, Office of Naval Research
(ONR) under the agreement N0001418-1-2068 and by the National Science
Foundation (NSF) under Grant No ECCS-1809234.}}

\begin{document}
%
\maketitle
\begin{abstract}
\textbf{In this paper, we analyze the non-convex framework of Wirtinger Flow (WF) 
for phase retrieval and identify a novel sufficient condition for universal exact recovery through the lens of low rank matrix recovery theory. 
Via a perspective in the lifted domain, we show that the convergence of the WF iterates to a true solution is attained geometrically under a single condition on the lifted forward model.
As a result, a deterministic relationship between the accuracy of spectral initialization and the validity of \emph{the regularity condition} is derived. 
In particular, we determine that a certain concentration property on the spectral matrix must hold uniformly with a sufficiently tight constant.
This culminates into a sufficient condition that is equivalent to a restricted isometry-type property over rank-1, positive semi-definite matrices, and amounts to a less stringent requirement on the lifted forward model than those of prominent low-rank-matrix-recovery methods in the literature. 
We characterize the performance limits of our framework in terms of the tightness of the concentration property via novel bounds on the convergence rate and on the signal-to-noise ratio such that the theoretical guarantees are valid using the spectral initialization at the proper sample complexity. 
}
\end{abstract}
\section{Introduction}
\label{sec:intro}

\subsection{Phase Retrieval}
Generalized phase retrieval (GPR) is a ubiquitous problem in science and engineering. 
The problem consists of the recovery of an object of interest $\x \in \mathbb{C}^N$ given the intensity only measurements of the form:
\begin{equation}\label{eq:phaless}
y_m = | \langle \a_m, \x \rangle |^2, \quad m = 1, 2, \cdots M, 
\end{equation} 
where $\a_m \in \mathbb{C}^N$ denotes the $m^{th}$ sampling vector.  
In literature, 
$\{ \a_m \}_{m = 1}^M$ most prominently corresponds to models such as Gaussian sampling \cite{candes2015phase}, coded diffraction patterns \cite{Candes13b}, or the rows of a known linear transformation, such as the short time Fourier transform \cite{bendory2018non}, or a particular imaging operator \cite{Chai11}. 
These models arise in problems such as X-Ray crystallography, coded diffraction imaging, optical astronomy,  quantum state tomography, array imaging, or blind channel estimation.  

One approach to address GPR is via a least-squares formulation in which the $\ell_2$ loss over intensity measurements in \eqref{eq:phaless} is minimized as follows:
\begin{equation}\label{eq:obj_func}
\underset{\z}{\text{minimize}} \ f(\z) := \frac{1}{2M} \sum_{m = 1}^M ( y_m - | \langle \a_m , \z \rangle |^2 )^2.
\end{equation}
Alternative forms of $f$ have also been popular in practice for optical imaging applications, where the $\ell_2$ loss is computed as a mismatch of amplitudes instead of intensities \cite{zhang2016reshaped}.
In solving GPR by \eqref{eq:obj_func} or using the amplitudes as measurements, the objective function $f$ is non-holomorphic and non-convex due to its invariance to global phase factors on the complex valued variable $\z$. 
Conventional methods from optical imaging literature reformulate \eqref{eq:obj_func} as a bilinear inverse problem by inserting the missing phase component as a variable, which is then solved by alternating minimization \cite{GS1972practical, fienup1978reconstruction}, or non-convex analogs of feasibility problems \cite{bauschke2003hybrid}. 
However, these methods are not equipped with practical recovery guarantees, and carry the risk of getting stuck in local minima due to the non-convexity of the problem.

Despite the ill-posed nature of the problem, there has been a significant progress in the development of provably good GPR algorithms in the last decade. 
Such methods are characterized by either one or both of the following two principles: convexification of the equality constraints and the solution set, which include \emph{lifting} based approaches \cite{Chai11, Candes13a, Candes13b, Waldspurger2015}, or a provably accurate initialization, followed by an algorithmic map that refines the initial estimate on the original signal domain \cite{netrapalli2013phase, candes2015phase, goldstein2018phasemax}. 
Notably, lifting-based approaches reformulate inversion from the quadratic equations of the form \eqref{eq:phaless} into a convex semi-definite program while squaring the dimension of the inverse problem. 
As a result, these solvers have demanding implementation costs due to computational complexity and memory requirements, which limit their applicability for large scale sensing problems. 
Essentially, methods that operate on the original signal domain evade such practical bottlenecks arising from the increased dimensionality of the inverse problem. 

\subsection{Wirtinger Flow}
The latter two-step approach for exact phase retrieval on the original signal domain was most prominently popularized by the seminal Wirtinger Flow (WF) framework \cite{candes2015phase}. 
In contrast to other state-of-the art exact phase retrieval methods, i.e., lifting or linear programming based approaches \cite{goldstein2018phasemax, hand2016elementary, bahmani2017flexible}, WF solves the original non-convex problem in \eqref{eq:obj_func} directly.

Given an initial estimate $\z_0$, WF performs steepest descent iterations by means of Wirtinger derivatives of ${f}$ as follows:
\begin{equation}\label{eq:WF_Updates}
\z_{k+1} = \z_{k} - \frac{\mu_{k}}{\| \z_0 \|^2} \nabla f(\z_{k}),
\end{equation}
where $\nabla f$ is defined as the complex gradient operator, and $\mu_{k}$ is the step size. 
The premise of WF is that if $\z_0$ is \emph{sufficiently accurate}, the iterates formed by \eqref{eq:WF_Updates} provably converge with a geometric rate to an element in the \emph{global solution set} which is defined as follows:

\begin{definition}{\emph{Global Solution Set}.}
Let
\begin{equation}\label{eq:SolutionSet}
\mathit{P} := \{ e^{j \phi} \x : \phi \in [0, 2\pi) \},
\end{equation}
where $\x \in \mathbf{C}^N$ is the ground truth of intensity measurements \eqref{eq:phaless}. The set $\mathit{P}$ is said to be
the global solution set of \eqref{eq:obj_func}.
\end{definition}

In general for any $\z \in \mathbb{C}^N$, the non-convex set of the form $\{ e^{j \Phi} \z : \Phi \in [0, 2\pi] \}$ represents an equivalence under the mapping of intensity only measurements. 
The convergence of algorithm iterates is governed by the following distance metric: 

\begin{definition}\label{def:DistMet}
Let $\x \in \mathbb{C}^N$ be an element of the solution set $\mathit{P}$. The distance of an element $\z \in \mathbb{C}^N$ to $\x$ is defined as \cite{candes2015phase}:
\begin{equation}\label{eq:distance}
\text{dist}(\z, {\x}) = \underset{\phi \in [0,2\pi]}{\text{arg min}} \ \| \z - {\x} e^{\mathrm{j} \phi} \|.
\end{equation}
The angle $\hat{\phi}$ where the minimum is achieved for a given $\z \in \mathbb{C}^N$ is denoted as $\Phi(\z)$. 
\end{definition}

In literal terms, \eqref{eq:distance} quantifies the distance of an estimate to the closest point in $\mathit{P}$, eliminating the effect of non-uniqueness caused by the global phase factors. 
As such, the exact phase retrieval refers to the iterates converging to any of the elements in the global solution set. 

Having to solve a non-convex problem, exact recovery guarantees of WF framework depend on the accuracy of the initial estimate $\z_0$ which is computed by the \emph{spectral method} \cite{netrapalli2013phase} as follows:
\begin{equation}\label{eq:spectral}
\mathbf{Y} = \frac{1}{{M}}\sum_{m = 1}^{{M}} y_{m} \mathbf{a}_m \mathbf{a}_m^H.
\end{equation}
The leading eigenvector of $\mathbf{Y}$, denoted as $\mathbf{v}_0$, is scaled by the square root of the normalized $\ell_1$-norm of the data, i.e., $\lambda_0 = M^{-1} \| \y \|_1$ to yield the initial estimate 
$\z_0 = \sqrt{\lambda_0} \mathbf{v}_0.$
Denoting $\text{dist}(\z_0, \x) = \epsilon \| \x \|$, the initial estimate determines an $\epsilon$-neighborhood of $\mathit{P}$ as follows:

\begin{definition}{\emph{$\epsilon$-Neighborhood of $\mathcal{P}$}.}
Let
\begin{equation}
\mathit{E}(\epsilon) = \{ \z \in \mathbb{C}^N : \text{dist}( \z, \mathit{P} ) \leq \epsilon \| \x \| \},
\end{equation}
where $\mathit{P}$ is the global solution set as defined in \eqref{eq:SolutionSet}. The set $\mathit{E}(\epsilon)$ is said to be the $\epsilon$-neighborhood of $\mathit{P}$.
\end{definition}

The main result of WF framework is that for Gaussian sampling and coded diffraction patterns, the initial estimate computed by the {spectral method} yields a small enough relative distance-$\epsilon$, such that the following \emph{regularity condition} holds with high probability for $M = \mathcal{O}(N \log N)$. 

\begin{condition}{\emph{Regularity Condition}.}\label{con:RegCon}
The objective function $f$ in \eqref{eq:obj_func} satisfies the regularity condition if, for all $\z \in \mathit{E}(\epsilon)$ the following inequality holds
\begin{equation}\label{eq:RegularityCond}
\text{\emph{Re}} \left( \langle \nabla {f}(\z),  (\z - {\x} e^{\mathrm{i} \Phi(\z)}) \rangle \right) \geq \frac{1}{\alpha} \text{\emph{dist}}^2 (\z, {\x}) + \frac{1}{\beta} \|  \nabla {f}(\z) \|^2
\end{equation}
with fixed $\alpha > 0$ and $\beta > 0$ such that $\alpha \beta > 4$. 
\end{condition}

Lemma 7.10 in \cite{candes2015phase} establishes that if the regularity condition is satisfied, the WF iterations are contractions with respect to the distance metric in \eqref{eq:distance} and all the algorithm iterates remain in $\mathit{E}(\epsilon)$. 
Essentially, the validity of \eqref{eq:RegularityCond} ensures that there exists no first order optimal point $\z \in \mathit{E}(\epsilon)$ other than the elements of $\mathit{P}$.

\subsection{Related Work and Our Contributions}

%

WF inspired several variants \cite{chen2017solving, chen2017, Zhang2017a, wang2017solving, wang2018phase, wang2018solving}, which improve on its performance guarantees with respect to computational and sample complexity, as well as robustness to noise and outliers.
The original WF framework has a $\mathcal{O}(1/N)$ specification on the algorithm step-size, yielding $\mathcal{O}(MN^2 \log 1/\epsilon_0)$ computational complexity for an $\epsilon_0$-relative accuracy on the final estimate of the optimization.
Furthermore, the sample complexity of the method is $M = \mathcal{O}(N \log N)$ in order to control the heavy tailed distribution of the spectral matrix, which also relates to the local curvature of the objective function within $E(\epsilon)$. 
These issues are largely mitigated by the Truncated WF (TWF) framework \cite{chen2017solving}, which guarantees convergence with a $\mathcal{O}(1)$ step-size for a linear computational complexity, and uniformly controls the tails of the spectral matrix when $M = \mathcal{O}(N)$ by devising a particular sample truncation scheme. 
Similar outcomes are achieved by alternative approaches that \emph{reshaped} the non-convex objective function to a loss over amplitudes \cite{zhang2016reshaped, wang2017solving, wang2018phase, wang2018solving} in the form of reduced sample complexity, and faster convergence.
Truncation methods were further studied in relation to increasing the robustness of these methods through median truncation \cite{Zhang2017a} or noise estimation \cite{chen2017}. 
Spectral initialization schemes were also subject to further studies such as those involving the design of optimal pre-processing functions \cite{lu2017phase, mondelli2017fundamental, luo2019optimal, ghods2018linear, gao2017phaseless}, generalizations \cite{bhojanapalli2016dropping, bariscan2018}, or alternative formulations including highest correlation estimators \cite{wang2018phase}, and orthogonality promoting methods \cite{wang2018solving}.

The aforementioned WF-inspired works offer exact recovery guarantees for phase retrieval based on a wide range of theoretical arguments which are prominently probabilistic in nature, derived through the properties of statistical models assumed for the underlying measurement maps.
Namely, the exact recovery analysis of state-of-the-art frameworks focus on establishing the regularity condition in \eqref{eq:RegularityCond} at the proper sample complexity, given that $E(\epsilon)$ is constructed by a specific initialization method.
We instead take a low-rank matrix recovery based approach to phase retrieval, and conduct a geometric analysis of the optimization problem in \eqref{eq:obj_func} for arbitrary measurement models. 

To this end, we develop a theoretical framework that unifies the key arguments that contribute to the exact recovery guarantees of WF and its intensity loss based variants under a single sufficient condition. 
Specifically, we show that one arrives at the restricted strong convexity property of the objective function around a global solution directly through a concentration bound of the spectral matrix due to the special structure of the set of rank-1, positive semi-definite (PSD) matrices.
We reach this conclusion by interpreting WF in the \emph{lifted domain}.
As a result, our framework establishes that the two steps of the non-convex optimization framework blueprinted by the seminal work in \cite{candes2015phase}, i.e., the accuracy of spectral initialization, and the regularity condition, are geometric outcomes of a less restrictive sufficient condition related to the following concentration bound:
\begin{equation}\label{eq:WF_init_intro}
\| \mathbf{Y}  - ( \mathbf{x} \mathbf{x}^H + \| \mathbf{x} \|^2 \mathbf{I} ) \| \leq \delta \| \x \|^2, \quad \text{for any} \ \x \in \mathbb{C}^N.
\end{equation}
\eqref{eq:WF_init_intro} is by no means an unexpected property. In fact, \eqref{eq:WF_init_intro} is known to hold true with high probability for Gaussian sampling and coded diffraction patterns when $M = \mathcal{O}(N \log N)$ through the concentration of the Hessian of $f$ around its expectation when evaluated at a global solution. 
Typically, it is used within the probabilistic analysis conducted for statistical models in relating the distance of the spectral initialization to the ground truth.
Its uniformity over all $\x \in \mathbb{C}^N$ is also established when $M = \mathcal{O}(N)$ under the sample truncation scheme of TWF \cite{chen2017solving}. 
In our work, we show that this is a much stronger property than it is given credit for. 
Namely, we prove that if the concentration bound in \eqref{eq:WF_init_intro} holds uniformly over all $\x$ where $\delta$ is sufficiently tight with $\delta \leq 0.184$, then the regularity condition is redundant for the exact recovery guarantees of WF starting from the spectral initialization. 
In other words, there surely exists positive $\alpha, \beta$ with $\alpha \beta > 4$ in Condition \ref{con:RegCon} such that \eqref{eq:RegularityCond} is satisfied \emph{deterministically} via the restricted strong convexity of the objective function in \eqref{eq:obj_func}. 

The resulting deterministic convergence framework amounts to two key results in this paper. 
First, under the validity of our sufficient condition, we identify the best achievable convergence rate that guarantees the exact recovery of any unknown from its intensity-only measurements. 
The upper bound on the convergence rate is determined solely by the concentration bound parameter $\delta$, which facilitates the derivation of an optimal fixed step-size $\mu_k = \mu(\delta)$ for the algorithm that is an $\mathcal{O}(1)$ constant. 
Secondly, in the presence of additive noise on the received intensity-only measurements, we determine a $\delta$-dependent lower bound for the signal-to-noise ratio such that the algorithm is guaranteed to yield an estimate that is within a bounded perturbation from the ground truth on expectation. 
Essentially, these results characterize the relationship of value of $\delta$ to the trade-offs between model parameters and performance of the algorithm with respect to convergence and stability for practical purposes, provided that $\delta \leq 0.184$.

Related to our theory of WF as a non-convex optimization framework, it is observed in \cite{chen2017solving} and \cite{sun2018geometric} that the regularity condition can be enforced by the restricted strong convexity condition due to the local Lipschitz differentiability of the objective function. 
Several key insights are built in \cite{sun2018geometric} for exactly solving \eqref{eq:obj_func}, in which a benign geometry for the objective function $f$ is realized with high probability when the number of samples are sufficiently large as $\mathcal{O}(N \log^3 N)$ for the Gaussian sampling model. 
In particular, it is observed that all local minimizers of $f$ are the elements of the global solution set $\mathit{P}$, and all its saddle points have a directional negative curvature with high probability, which allow vanilla gradient descent to converge to the exact solution even if the algorithm is initialized randomly. 
Through the implicit regularization properties of gradient descent and the incoherence property in the Gaussian sampling model, \cite{ma2019implicit} improves the step size of WF to $\mathcal{O}(1/ \log N)$ starting from spectral initialization when $M = \mathcal{O}(N \log N)$, thereby attaining a near linear computational complexity of $\mathcal{O}(MN \log N \log 1/\epsilon_0)$. 
\cite{ma2019implicit} also identifies a $\mathcal{O}(N \log^3 N)$ sample complexity for convergence of WF with a constant step-size. 
In contrast, \cite{sanghavi2017local} establishes the restricted strong convexity of the objective function \eqref{eq:obj_func} around the solution set at $\mathcal{O}(N \log^2 N)$ sample complexity for the Gaussian sampling model, with spectral initialization proven to fall within the strongly convex region with high probability.
In this paper, we improve upon the required sample complexity identified in \cite{sanghavi2017local} to $\mathcal{O}(N \log N)$ for the restricted strong convexity of \eqref{eq:obj_func} in a sufficiently wide neighborhood around the solution set that is guaranteed to include the initial spectral estimate.
Furthermore, our approach facilitates reducing the computational complexity of the algorithm to $\mathcal{O}(MN \log 1/\epsilon_0)$ for $\epsilon_0$-relative accuracy due to a constant $O(1)$ step-size, which is an improvement over \cite{candes2015phase} and \cite{ma2019implicit} at the identical sampling complexity.

A notable work in the literature in relation to our results is \cite{bhojanapalli2016dropping}.
When coupled with the result in \cite{sanghavi2017local}, \cite{bhojanapalli2016dropping} guarantees the exact recovery of low rank-$r$ matrix deterministically via restricted strong convexity with a properly set step-size that relates to the spectra of the ground truth. 
Another work that closely relates to our framework is \cite{li2018rapid}, in which a non-convex approach based on a local restricted isometry property (RIP) over rank-2 matrices is considered as sufficient and established for the Gaussian sampling model with high probability for exactly solving the blind deconvolution problem. 
In essence, our framework stands in agreement with \cite{bhojanapalli2016dropping, li2018rapid}. 
Specifically, the sufficient condition we identify on the lifted forward model for exact non-convex phase retrieval \emph{i}) deterministically yields a local RIP-2 type condition reminiscent to that of \cite{li2018rapid} when specified to hold only over the difference of two rank-1, PSD matrices, and \emph{ii}) directly implies restricted strong convexity due to the duality between lifted domain and the underlying signal domain, after which results of \cite{bhojanapalli2016dropping} become applicable. 
Furthermore, under the validity of the concentration bound over all $\x \in \mathbb{C}^N$, the sufficient condition for exact recovery is effectively converted to a RIP on the lifted forward model of the problem over the set of rank-1, positive semi-definite (PSD) matrices.
In this manner, the developed framework realizes the full characterization of an observation made in \cite{bhojanapalli2016dropping} by deriving the best uniform convergence rate and the optimal step-size for exact recovery only with respect to the restricted isometry constant (RIC). 
Thereby, our framework facilitates a bridge between the optimization based low rank matrix recovery methods, and the prominent statistical frameworks for phase retrieval with a unifying sufficient condition. 


In recent years, RIP-based conditions on the lifted measurement map have been established as sufficient for global optimality in recovering a rank-$r$ matrix from quadratic or bilinear equations via first-order non-convex optimization methods.
RIP over rank-$6r$ with a RIC less than or equal to $1/10$ is shown to be sufficient for exact recovery via Procrutes Flow \cite{tu2015low}.
RIP over rank-$4r$ matrices with a RIC less than or equal to $1/5$ is shown to guarantee the strict saddle point condition and the absence of any spurious local minima for the $\ell_2$ loss function \cite{zhu2018global}.
RIP over rank-$2r$ matrices with a RIC less than or equal to $1/5$ is shown to be sufficient for having no spurious local minima \cite{bhojanapalli2016global} whereas at most a RIC of $1/2$ over rank-$2r$ matrices is postulated as necessary in \cite{zhang2018much}. 
Notably, the sufficient condition we establish in our framework corresponds to a less restrictive RIP-type condition than those of the non-convex low rank matrix recovery methods in \cite{tu2015low, zhu2018global, bhojanapalli2016global, zhang2018much}, as it suffices that the property is satisfied only over the set of rank-1, PSD matrices. 
The major difference of our result stems from the fact that we are merely interested in restricted strong convexity within the $\epsilon$-neighborhood obtained from the spectral initialization, and not the global properties of the optimization landscape. 
As a result, to the best of our knowledge, with this paper we establish the most {minimal} RIP based framework that fully characterizes the performance guarantees of non-convex phase retrieval via WF. 
Other key works for the non-convex rank-$r$ matrix recovery problem include \cite{zheng2015convergent}, in which the regularity condition of WF is considered and shown to hold for the Gaussian sampling model, and \cite{bhojanapalli2016dropping}, \cite{wang2016unified} in which exact recovery, and stability are studied under restricted strong convexity and smoothness of the objective function, respectively. 
For further discussion on advances in non-convex low rank matrix recovery, we refer the reader to \cite{chi2018nonconvex}.

\subsection{Notation and Organization of the Paper}
The rest of the paper is organized as follows. 
In Section \ref{sec:WF}, we provide a preliminary discussion on the interpretation of WF in the lifted domain.
Section \ref{sec:Results} contains our main results, and remarks. 
Section \ref{sec:Rob} evaluates the robustness of WF in the presence of additive noise. 
In Section \ref{sec:Apdx}, we present the proofs of our results. Section \ref{sec:Conc} concludes the paper.  

We denote the elements of finite dimensional vector spaces with lower case bold letters. 
Upper case bold and italic letters are allocated for matrices and sets, respectively. Caligraphic letters are allocated for operators that act on the lifted domain in $\mathbb{C}^{N \times N}$. In denoting the norms of elements in different domains, we use the following notation: $\| \cdot \|$ denotes the $\ell_2$ norm when acting on a vector, and the spectral norm when acting on an operator. $\| \cdot \|_F$ and $\| \cdot \|_*$ denote the Frobenius and nuclear norms, respectively. $\mathbf{I}$ and $\mathcal{I}$ denote the identity operators on the vector space of $\mathbb{C}^N$, and the lifted domain in $\mathbb{C}^{N \times N}$, respectively.

\section{WF in the Lifted Domain}
\label{sec:WF}
We start by interpreting WF as a solver in the \emph{lifted domain}, and adopt the concepts of the seminal work of PhaseLift in \cite{Candes13b}.  
Lifting based approaches provide a profound perspective to the phase retrieval problem.
In principle, these methods target the core issue of non-injectivity of phaseless measurement maps, which is a key step in formulating methods that guarantee exact recovery in phase retrieval literature \cite{bandeira2014saving, goldstein2018phasemax}. 
Notably, one can consider the measurement model in \eqref{eq:phaless} as a mapping from a rank-1, positive semi-definite matrix $\x \x^H \in \mathbb{C}^{N \times N}$ instead of a quadratic map from the signal domain in $\mathbb{C}^N$.
Lifting conceptualizes this observation:
\begin{definition}{\emph{Lifting}.}
Each measurement in \eqref{eq:phaless} can be expressed in the form of an inner product of two rank-1 operators, $\tilde{\X} = \x \x^H$ and ${\mathbf{A}}_m = \mathbf{a}_m \mathbf{a}_m^H$ such that
\begin{equation}\label{eq:PhaseLift}
y_m = \langle {\mathbf{A}}_m , \tilde{\X} \rangle_F \quad m = 1, . . ., M
\end{equation}
where $\langle \cdot, \cdot \rangle_F$ is the Frobenius inner product. The process of transforming the signal recovery over $\mathbb{C}^N$ to the recovery of the rank-1 unknown $\tilde{\X} \in \mathbb{C}^{N \times N}$ is known as {lifting}. 
\end{definition}

The lifting technique introduces a new, linear measurement map $\mathcal{A}: \mathbb{C}^{N \times N} \rightarrow \mathbb{C}^M$, which we refer to as \emph{the lifted forward model}. 
Specifically, for the phaseless measurement model in \eqref{eq:phaless}, the domain of $\mathcal{A}$ is constrained on the set of rank-1, positive semi-definite (PSD) matrices $\mathit{R}^{+}_1 = \{ \z \z^H : \z \in \mathbb{C}^N \}$ as follows:
\begin{equation}\label{eq:LiftFM}
\mathbf{y} = \mathcal{A}(\x \x^H)
\end{equation}
where $\mathbf{y} = [y_1, y_2, \cdots y_M] \in \mathbb{R}^M$. 
As a result, each non-convex set of equivalent points under the mapping from the signal domain in $\mathbb{C}^N$ to the phaseless measurements, i.e., $\{ \z \mathrm{e}^{\mathrm{i} \Phi}: \Phi \in [0, 2\pi) \}$ for $\z \in \mathbb{C}^N$, is \emph{compressed} into a single element $\z \z^H  \in \mathit{R}^{+}_1$.  
Thereby, quadratic equality constraints over the signal domain are transformed to affine equality constraints in the lifted domain in $\mathbb{C}^{N \times N}$, which define a convex manifold. 

In typical inference problems, $\mathcal{A}$ has a non-trivial null space as the system of linear equations in \eqref{eq:PhaseLift} is severely underdetermined with $M \ll N^2$. 
Various studies approach the phase retrieval problem over the lifted domain, leveraging the low rank structure of the lifted ground truth $\tilde{\X} = \x \x^H$ and the subsequent LRMR theory from compressed sensing and matrix completion literature \cite{cai2010singular, recht2010guaranteed}.  
The sufficient conditions on $\mathcal{A}$ for exact recovery of $\x \x^H$ are primarily characterized by its null space \cite{recht2008necessary, recht2011null, oymak2011simplified} or restricted isometry properties on low rank \cite{recht2010guaranteed, cai2013sharp, bhojanapalli2016global} or PSD \cite{Candes13b} matrices.


Knowing that \eqref{eq:obj_func} corresponds to the minimization of an $\ell_2$ loss objective, WF exclusively iterates on the set of rank-1, PSD matrices by solving the following:
\begin{equation}\label{eq:WFform}
\underset{\mathbf{X}}{\text{minimize: }} \ \frac{1}{2M} \| \mathcal{A} ( \mathbf{X} ) - \mathbf{y} \|^2 \quad \text{s.t.} \quad \mathbf{X} \in\mathit{R}^{+}_1,
\end{equation}
where 
$\mathbf{X}$ denotes the optimization variable in the lifted domain. 
The functional constraint on $\mathbf{X}$ as rank-1, PSD matrix casts this minimization equivalent to minimizing over the signal domain variable $\z$, resulting with dimensionality reduction of the search space. 
This is practically enforced by a spectral projection within the gradient term $\nabla f$, which can be expressed as
\begin{equation}\label{eq:GradEq}
\nabla f(\z) = \frac{1}{M}\mathcal{A}^H \mathcal{A} (\z \z^H - \x \x^H) \z.
\end{equation}

Beyond the immediate gains in practicality, the formulation in \eqref{eq:WFform} reveals a theoretical advantage offered by the non-convex framework of WF, using which deterministic arguments for exact recovery similar to those of lifting-based methods can also be attained \cite{bariscan2018}. 
Moving from the convex relaxations of rank-minimization, WF corresponds to solving a non-convex feasibility problem, reminiscent of the optimizationless PhaseLift method in \cite{Demanet14}, and Uzawa's iterations in \cite{cai2010singular}. 
This yields an iterative scheme for the unrelaxed, non-convex form of the lifted problem, and enforces the rank-1, PSD structure on the iterates. 
Thereby, the constraint set of WF is considerably less restrictive than those of the prominent LRMR methods that solve for the lifted unknown $\x \x^H$. 

Namely, the problem in \eqref{eq:WFform} has a unique solution if there exists no element $\mathbf{H}$ in the null space of $\mathcal{A}$, such that $\x \x^H + \mathbf{H}$ is a rank-1, PSD matrix.
Observe that the validity of Condition \ref{con:RegCon} and non-uniqueness of $\mathcal{A}$ over the set of rank-1, PSD matrices result in a \emph{contradiction}.
Essentially, existence of a $\bpsi \notin \mathit{P}$ that satisfies \eqref{eq:phaless} ascertains either that the spectral initialization $\z_0$ is at a stationary point itself, or that the $\epsilon$-neighborhood around the $\bpsi$ or $P$ contains a stationary point, both of which violate the regularity condition. 
Hence, the degree to which $\mathbf{Y}$ approximates $\x \x^H$ in the spectral norm sense as \eqref{eq:WF_init_intro}, also dictates the feasibility of violating the uniqueness condition of $\mathcal{A}$ over the set of rank-1, PSD matrices. 
This is made further clear under the following observation. 
\begin{remark}
The spectral matrix $\mathbf{Y}$ in \eqref{eq:spectral} is the \emph{backprojection} estimate of the lifted ground truth, $\tilde{\X} = \x \x^H$, i.e.,
\begin{equation}\label{eq:WF_init2}
\mathbf{Y} = \frac{1}{M} \mathcal{A}^H (\mathbf{y} ),
\end{equation}
which, in the noise free case is $\mathbf{Y} = \frac{1}{M} \mathcal{A}^H \mathcal{A}(\x \x^H)$. 
\end{remark}

In practical terms, \eqref{eq:WF_init_intro} becomes a condition on the \emph{lifted normal operator}, $\mathcal{A}^H \mathcal{A}$. 
Thereby, the main result of this paper is that, if the following \emph{concentration bound}
\begin{equation}\label{eq:WF_init}
\| \frac{1}{M}\mathcal{A}^H \mathcal{A}(\x \x^H)  - ( \mathbf{x} \mathbf{x}^H + \| \mathbf{x} \|^2 \mathbf{I} ) \| \leq \delta \| \x \|^2,
\end{equation}
holds for all $\x \in \mathbb{C}^N$ 
with a $\delta$ that is sufficiently small, then via the special structure of the set of rank-1, PSD matrices, the iterations in \eqref{eq:WF_Updates} are guaranteed to converge to a solution in $\mathit{P}$ via the restricted strong convexity of $\mathcal{J}$ in $\mathit{E}(\epsilon)$. 
As $\delta$ gets smaller, the spectral initialization yields more accurate estimates due to favorable properties of the lifted normal operator over the set of rank-1, PSD matrices. 
Due to the fact that shrinkage on the value of $\delta$ is related to increasing the number of measurements $M$, the existence of sub-optimal minima accordingly vanishes. 
Hence, there exists a phase transition with respect to the value of $\delta$, below which the tightness of the concentration bound can deterministically guarantee exact recovery from \eqref{eq:phaless} using WF. 




\section{Main Results}
\label{sec:Results}


In 
this section, 
we prove \eqref{eq:WF_init} as a sufficient condition for exact phase retrieval for an arbitrary measurement model by establishing a set of lemmas through fully geometric arguments, given that it holds for all $\x \in \mathbb{C}^N$ with $\delta < 0.184$.
Thereby, we use the form in \eqref{eq:WF_init2} as a crucial element of our approach. 
Note that the lifted forward model in \eqref{eq:LiftFM} may be a realization from a statistical model, or a deterministic measurement map.
The following lemma characterizes the normal operator of the lifted forward model over the set of rank-1, positive semi-definite matrices. 


\begin{lemma}\label{lem:Lemma1}
Assume that \eqref{eq:WF_init} holds for any $\x \in \mathbb{C}^N$. Then, the normal operator of the lifted forward model can be expressed as follows over the set of rank-1, PSD matrices:
\begin{equation}
\frac{1}{M} \mathcal{A}^H \mathcal{A} = \mathcal{I} + \mathcal{R} + {\Delta},
\end{equation}
where for any ${\z \z^H} \in \mathit{R}^{+}_1$, $\mathcal{R}({\z} \z^H) = \| \z \|^2 \mathbf{I}$, with
\begin{equation}
\mathcal{R}({\z \z^H} - {\x \x^H}) = (\| \z \|^2 - \| \x \|^2) \mathbf{I},
\end{equation}
and ${\Delta}: \mathit{R}^{+}_1 \rightarrow \mathbb{C}^{N \times N}$ is a perturbation operator satisfying $\| {\Delta}({\x \x^H}) \| \leq \delta \| \x \|^2$ for any $\x \in \mathbb{C}^N$ such that
\end{lemma}
\begin{equation}\label{eq:DefsSpec}
\underset{\mathbf{v \in \mathbb{C}^N \setminus \{0 \}}}{\text{max}} \frac{\| {\Delta} (\v \v^H) \|}{\| \v \v^H \|} \leq \delta.
\end{equation}
\begin{proof}
See Section \ref{sec:Prf1} and Appendix \ref{sec:App}. 
\end{proof}

Specifically in the case of the Gaussian model, the operator $\mathcal{R}: \mathbb{C}^{N \times N} \rightarrow \mathbb{C}^{N \times N}$ characterizes the effect of the $4^{th}$ moments of the sampling vectors. 
This term captures the diagonal bias of the spectral matrix in estimating the lifted signal, using the fact that the expectation of the lifted normal operator is linear on $\mathbb{C}^{N \times N}$. 

We begin by considering the spectral initialization scheme. 
Namely, through Lemma \ref{lem:Lemma1}, the concentration bound in \eqref{eq:WF_init} indicates a proper scaling factor for the unit-norm eigenvector of $\mathbf{Y}$. 
This is derived in the following corollary which implies that the lifted forward model $\mathcal{A}$ is a \emph{tight frame}.
\begin{corollary}\label{cor:Corr3}
Assume that the assumptions of Lemma \ref{lem:Lemma1} hold. Then, $\mathcal{A}$ satisfies the following identity: 
\begin{equation}\label{eq:tightframe}
(2 - \delta) \| {\x \x^H} \|_F^2 \leq \frac{1}{M} \| \mathcal{A}(\x \x^H) \|^2 \leq (2 + \delta) \| {\x \x^H} \|_F^2.
\end{equation}
Furthermore, if \eqref{eq:tightframe} holds for all $\x \in \mathbb{C}^N$, then the concentration bound in \eqref{eq:WF_init} equivalently holds uniformly over all $\x \in \mathbb{C}^N$.
\end{corollary}
\begin{proof}
See Section \ref{sec:Prf1}.
\end{proof}

As a result, $\lambda_0 = (2M)^{-1/2} \| \y \|$ is an estimator for the energy of the signal $\| \x \|^2$. 
Using the norm estimate obtained by Corollary \ref{cor:Corr3} for $\lambda_0$, the distance of the spectral initialization yields the following $\epsilon$-neighborhood as a function of the concentration parameter $\delta$. 

\begin{lemma}\label{lem:Lemma2}
Assume that the assumptions of Lemma \ref{lem:Lemma1} hold. Let $\z_0$ be the estimate $\z_0 = \sqrt{\lambda_0} \v_0$ where $\v_0$ is the eigenvector corresponding to the leading eigenvalue of the spectral matrix $\mathbf{Y}$ in \eqref{eq:WF_init2}, and $\lambda_0$ is the signal energy estimate obtained as
\begin{equation}\label{eq:normest}
\lambda_0 = \frac{1}{\sqrt{2M}} \| \y \|. 
\end{equation}
Then, the initial estimate $\z_0$ satisfies $\mathrm{dist}^2(\z_0, \x) \leq \epsilon^2 \| \x \|^2$, where
\begin{equation}\label{eq:epsilonval}
\epsilon^2 = 1 + \sqrt{1 + \frac{\delta}{2}} - 2 \sqrt{\left(1 - 2 \delta\right) {\left(1 + \frac{\delta}{2}\right) }^{1/2}}.  
\end{equation}
\end{lemma}
\begin{proof}
See Section \ref{sec:Prf2}. 
\end{proof}

Note that $E(\epsilon)$ is formed using \eqref{eq:epsilonval} 
under the assumption that $\delta < 0.5$, which is required to have a valid estimate via spectral initialization. 
Next, we introduce the following lemma to characterize the relation between the distance metric introduced in \eqref{eq:distance}, and the distance in the lifted domain. 

\begin{lemma}\label{lem:Lemma3}
Let $ \z \in E(\epsilon)$ of $\x$, with $\epsilon = \epsilon_0 \| \x \|$ satisfying \eqref{eq:epsilonval}. Then, we have
\begin{equation}
h_1(\delta) \mathrm{dist}(\z, \x) \| \x \| \leq \| {\z \z^H} - {\x \x^H} \|_F \leq h_2(\delta) \mathrm{dist}(\z, \x) \| \x \|,
\end{equation}
where $h_1 = \sqrt{(1-\epsilon)(2-\epsilon)}$, and $h_2 =  (2+\epsilon)$. 
\end{lemma}
\begin{proof}
See Section \ref{sec:Prf3}.
\end{proof}

Lemma \ref{lem:Lemma3} states that the distance between the lifted signals ${\z \z^H}$ and ${\x \x^H}$ is of the rate of the distance of the signals $\z, \x \in \mathbb{C}^N$, when $\z \in \mathit{E}(\epsilon)$ of $\x$. 
Essentially, the distance metric in \eqref{eq:distance} locally tracks the error on the constraint set of rank-1, PSD matrices in the lifted domain. 
The outcome of this result, together with Lemma \ref{lem:Lemma1}, is the following \emph{local restricted isometry}-type property. 

\begin{lemma}\label{lem:Lemma4}
Under the assumptions of Lemmas \ref{lem:Lemma1} and \ref{lem:Lemma3}, for any $\x \in \mathbb{C}^N$ and $\z \in \mathit{E}(\epsilon)$, the lifted forward model satisfies
\begin{align}\label{eq:RIP2}
(1 - \hat{\delta} ) \| {\z \z^H} - {\x \x^H} \|_F^2 &\leq  \frac{1}{M}\| \mathcal{A}({\z} \z^H - {\x} \x^H) \|^2 \\
& \leq (2 + \hat{\delta}) \| {\z}\z^H - {\x} \x^H \|_F^2, \nonumber
\end{align}
where $ \hat{\delta} = \frac{\sqrt{2}(2+\epsilon)}{\sqrt{(1-\epsilon)(2-\epsilon})} \delta$. 
\end{lemma}
\begin{proof}
See Section \ref{sec:Prf4}.
\end{proof}


We refer to $\hat{\delta}$ as the {local} restricted isometry constant (RIC) of the lifted forward model over rank-2 matrices. 
However, note that this property holds over a very particular subset of rank-2 matrices even beyond the locality with respect to $\x \in \mathbb{C}^N$. 
Namely, \eqref{eq:RIP2} states that the distance between two elements $\z \z^H, \x \x^H \in \mathit{R}^{+}_1$ is approximately preserved under the mapping of $\mathcal{A}$, if $\z \in E(\epsilon)$ of $\x$. 
The significance of Lemma \ref{lem:Lemma4} is the fact that the restricted isometry property in \eqref{eq:RIP2} is derived as a deterministic consequence of the concentration bound of the spectral matrix. 
The Lemmas \ref{lem:Lemma3} and \ref{lem:Lemma4} culminate to yield the Lipschitz differentiability of the objective function, stated in the following lemma: 

\begin{lemma}\label{lem:Lemma5}
Assume that the assumptions of Lemmas \ref{lem:Lemma3} and \ref{lem:Lemma4} hold. Then, for any $\z \in E(\epsilon)$, the objective function $f$ in \eqref{eq:obj_func} is local Lipschitz differentiable at $\x \in \mathit{P}$ with
\begin{equation}
\| \nabla f(\z) \| \leq c(\delta) \cdot \mathrm{dist} (\z, \x) \| \x \|^2
\end{equation}
where $ c(\delta) = (1+\epsilon)(2+\epsilon) (2+{\delta})$ is the local Lipschitz constant. 
\end{lemma}
\begin{proof}
See Section \ref{sec:Prf5}. 
\end{proof}

Invoking the result of Lemma \ref{lem:Lemma5}, to establish the regularity condition for $f$, it is sufficient to show that for any $\z \in E(\epsilon)$
\begin{equation}
\label{eq:RegCond_new}
\mathrm{Re} \left( \langle \nabla f(\z),  (\z - \x e^{\mathrm{i} \phi(\z)}) \rangle \right) \geq (\frac{1}{\alpha} +  \frac{c^2(\delta) \| \x \|^4}{\beta}) \ \mathrm{dist}^2 (\z, \x)
\end{equation}
which is equivalent to the \emph{restricted strong convexity} of the objective function in $\mathit{E}(\epsilon)$.
By the definition of strong convexity around the closest solution $\hat{\x} = \mathrm{e}^{j \Phi(\z)} \x $ to an estimate $\z \in \mathit{E}(\epsilon)$, this condition, is implied if the objective function satisfies
\begin{equation}\label{eq:RestStongConv}
f(\z) \geq f(\hat{\x}) + \mathrm{Re}\left( \nabla f (\hat{\x})^H (\z - \hat{\x}) \right) + \frac{L}{2} \| \z - \hat{\x} \|^2,
\end{equation}
where $L$ equals to the multiplier of the distance term in \eqref{eq:RegCond_new}. 
Having $ f(\hat{\x})$ and $\nabla f (\hat{\x})$ equal $0$ by definition, the restricted strong convexity in $\mathit{E}(\epsilon)$ is simply reduced to the following condition:
\begin{equation}\label{eq:RegConditionFinal}
{f}(\z) \geq \frac{1}{2}\left({\frac{1}{\alpha} +  \frac{c^2(\delta)  \| \x \|^4}{\beta}}\right)  \mathrm{dist}^2 (\z, \x),
\end{equation}
for any $\z \in \mathit{E}(\epsilon)$.
Writing $f$ explicitly in terms of the lifted signals as $f(\z) = \| \mathcal{A} ({\z \z^H}- {\x \x^H})  \|^2 / 2M$, and applying the lower bounds from Lemmas \ref{lem:Lemma3} and \ref{lem:Lemma4}, we have
\begin{equation}
f(\z) \geq \frac{(1-\hat{\delta}) h_1^2(\delta)}{2}  \mathrm{dist}^2 (\z, \x) \| \x \|^2.
\end{equation}
Thus, the regularity condition is satisfied by setting $\alpha$ and $\beta$ such that $\alpha \beta > 4$, and 
\begin{equation}\label{eq:RegCond_Final}
\frac{1}{\alpha \| \x \|^2} + \frac{c^2(\delta) \| \x \|^2}{\beta} \leq (1-\hat{\delta}) h_1^2(\delta) : = h(\delta).
\end{equation}

The final form we derive in \eqref{eq:RegCond_Final} results in a number of notable outcomes regarding the non-convex optimization theory of the WF framework:

$1)$ We show that there exists a regime in which the regularity condition holds by default.
This regime is characterized by the concentration bound of the spectral matrix in \eqref{eq:WF_init}, as $\hat{\delta}$ is solely a function of $\delta$. 
Observe that the validity of this regime depends on attaining $\hat{\delta} < 1$ which constrains the tightness of the concentration property in \eqref{eq:WF_init}. 
This numerically yields an upper bound of $\delta \leq 0.184$ as shown in Figure \ref{fig:1}.

$2)$ \eqref{eq:RegCond_Final} provides an interpretation of the algorithm parameters consistent with the original work of \cite{candes2015phase}:
Figure \ref{fig:2} demonstrates the range of values the constants $c$ and $h$ can attain in the valid region of $\delta$. 
Notably, the values of these $\mathcal{O}(1)$ constants characterize the convergence rate of the algorithm, as $\alpha$ and $\beta$ are required to be sufficiently large constants for \eqref{eq:RegCond_Final} to hold. 
Observe that \eqref{eq:RegCond_Final} implies setting $\alpha = \mathcal{O}({1}/{\| \x \|^2})$, and $\beta = \mathcal{O}(\| \x \|^2)$, hence $\alpha \beta = \mathcal{O}(1)$. 
Since we clearly have $h < 2$, and $c > 4$, the condition in \eqref{eq:RegCond_Final} holds with $\alpha \beta > 4$ by definition. 
Hence, the regularity condition is satisfied, and a step size $\mu' \leq {2}/{\beta}$ can be chosen to yield a convergence rate of ${2\mu'}/{\alpha}$ via \cite{candes2015phase}.
This step size $\mu'$ is then $\mathcal{O}({1}/{\| \x \|^2})$. 
Hence, the definition of the updates in \eqref{eq:WF_Updates} requires an approximate normalization term\footnote{The approximation argument can be followed from \cite{candes2015phase} in the proof of Lemma 7.10.} of $\| \z_0 \|^2$ on a scalar entity $\mu_{k} = \mathcal{O}(1)$. 

\begin{figure}
\centering
\includegraphics[scale=0.3]{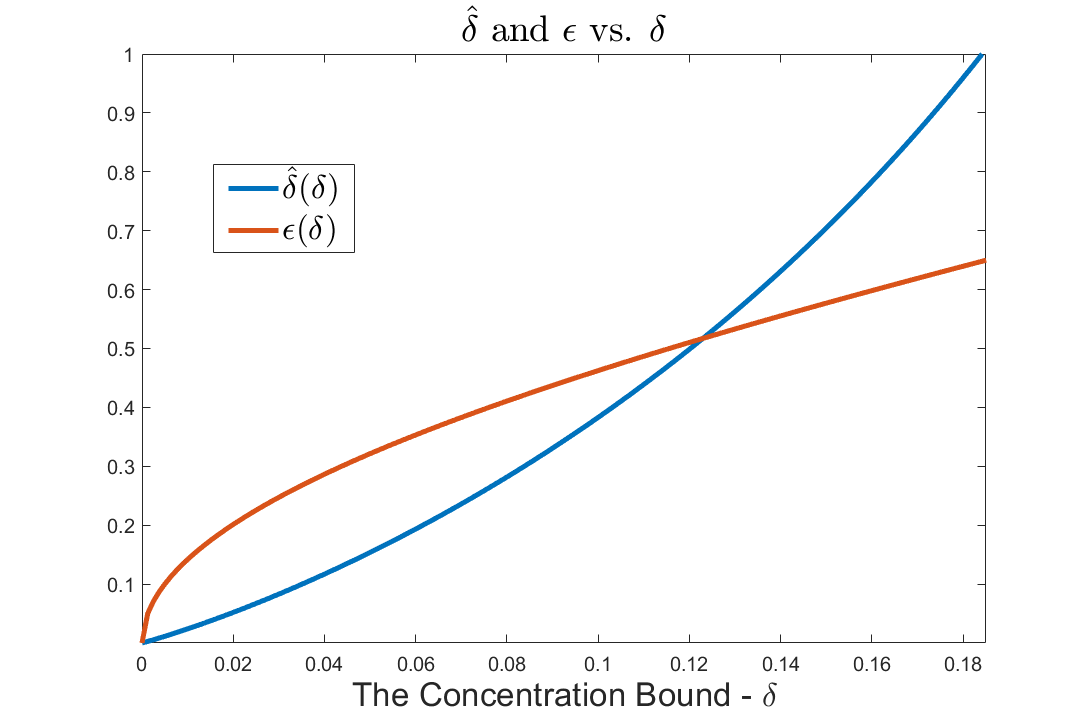}
\caption{\emph{The $\hat{\delta}$ and $\epsilon$ values in the region that the regularity condition is guaranteed to hold.} Observe that the limiting value is the local restricted isometry constant, $\hat{\delta}$, which controls the uniqueness property in the lifted problem locally.}
\label{fig:1}
\end{figure}

\begin{figure}
\centering
\includegraphics[scale=0.3]{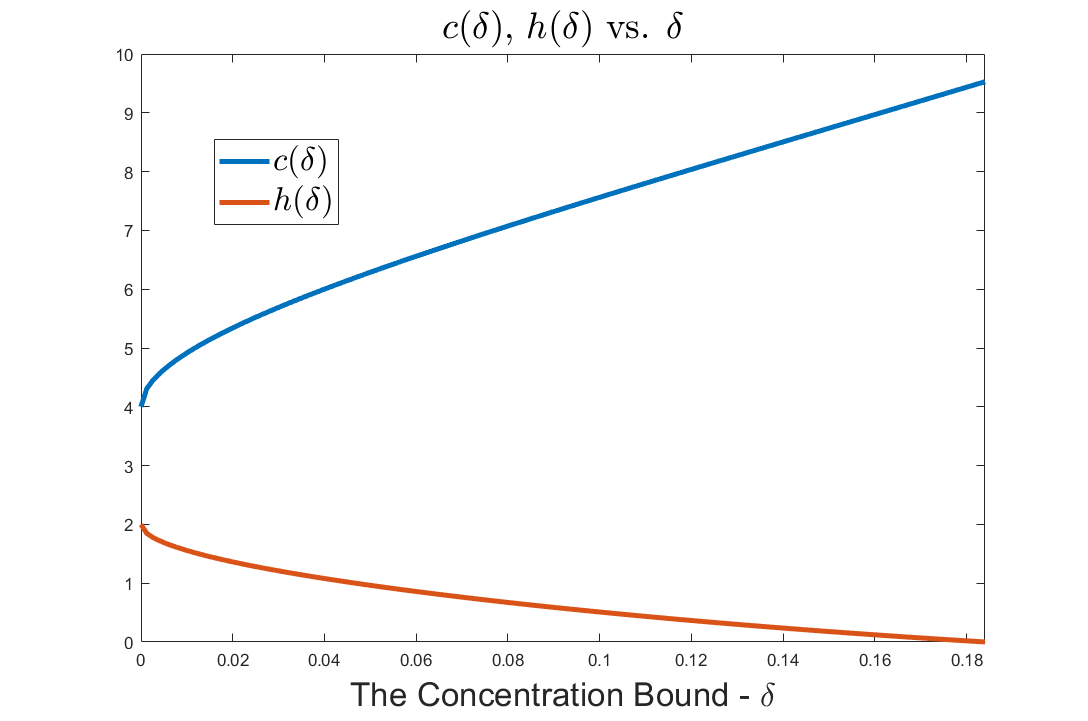}
\caption{\emph{The $c$ and $h$ values in the region that the regularity condition is guaranteed to hold.} The two constants characterize the values of $\alpha$ and $\beta$ parameters for the convergence of the iterates generated by WF updates.}
\label{fig:2}
\end{figure}

$3)$ The condition in \eqref{eq:RegCond_Final} effectively places an upper bound on the convergence rate of the algorithm under which exact recovery of any $\x \in \mathbb{C}^N$ is guaranteed. 
Simply fixing $\kappa = \alpha \beta$ and re-organizing \eqref{eq:RegCond_Final}, we obtain the following:
\begin{equation}\label{eq:Conv_Bound_1}
\frac{1}{\alpha \| \x \|^2} \left( 1 + \frac{c^2(\delta) \| \x \|^4}{\kappa} \alpha^2 - h(\delta) \| \x \|^2 \alpha \right) \leq 0. 
\end{equation}
Since we have $\alpha > 0$ by definition and $h(\delta) > 0$ by constraining $\delta$, it suffices to consider the non-negativity of the discriminant of the quadratic equation with respect to $\alpha$ in \eqref{eq:Conv_Bound_1} for the overall condition to hold, which yields
\begin{equation}
h^2(\delta) \| \x \|^4 - \frac{4}{\kappa} c^2(\delta) \| \x \|^4 \geq 0. 
\end{equation}
As a result, knowing that $4/\kappa$ is an upper bound on the rate of convergence of the algorithm via \cite{candes2015phase}, we obtain the best achievable geometric convergence rate as a function of the concentration bound parameter as
\begin{equation}\label{eq:ab}
\frac{4}{\alpha \beta} \leq \frac{h^2(\delta)}{c^2(\delta)} := r(\delta) = \left(\frac{(1-\hat{\delta})(1-\epsilon)(2-\epsilon)}{(2+{\delta})(1+\epsilon)(2+\epsilon)} \right)^2.
\end{equation}
Thereby, beyond directly guaranteeing the existence of a pair of $(\alpha, \beta)$ to satisfy the regularity condition when sufficiently small, the $\delta$-value fully characterizes the practicality and iteration complexity of the algorithm via $r(\delta)$ in \eqref{eq:ab}. 
Note that $r(\delta)$ is the best convergence rate the algorithm can achieve uniformly over all $\x \in \mathbb{C}^N$.
 
\begin{figure}
\centering
\includegraphics[scale=0.3]{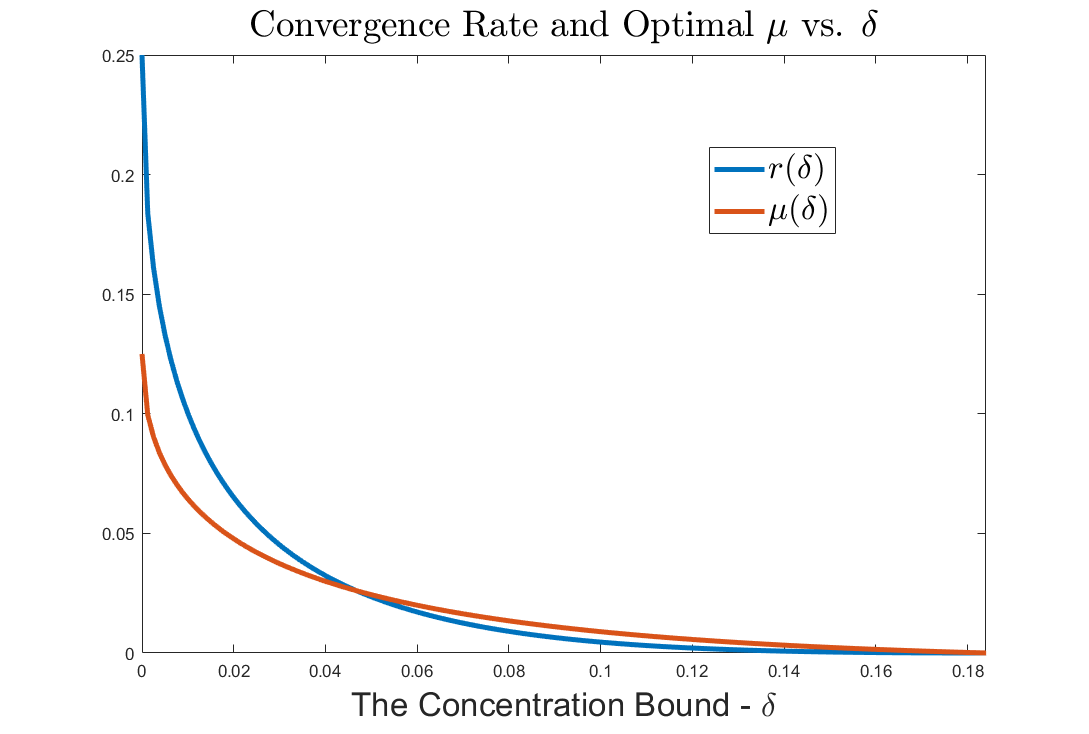}
\caption{\emph{The $r(\delta)$ upper bound on rate of converge and the optimal step size $\mu(\delta)$ with respect to the concentration bound parameter $\delta$ in the exact recovery regime of our framework.} The optimal step-size  is defined as the $\mu$ that satisfies the upper bound on the rate-of convergence.}
\label{fig:3}
\end{figure}

Organizing the arguments developed in this section, we state the following for exact phase retrieval via Wirtinger Flow:
\begin{theorem}\label{thm:Theorem1}
Assume that the concentration property
$$
\| \frac{1}{M} \mathcal{A}^H \mathcal{A}({\x} \x^H)  - ( \mathbf{x} \mathbf{x}^H + \| \mathbf{x} \|^2 \mathbf{I} ) \| \leq \delta \| \x \|^2. 
$$
holds for all $\x \in \mathbb{C}^N$ with $\delta \leq 0.184$. 
Then, the initial estimate $\z_0$ obtained from the spectral matrix in \eqref{eq:spectral} using the scaling factor in \eqref{eq:normest} satisfies
$$
\mathrm{dist}^2(\z_0, \x)  \leq \epsilon^2 \| \x \|^2, 
$$
for all $\x$, where $\epsilon^2 \leq 1 + \sqrt{1 + \frac{\delta}{2}} - 2 \sqrt{(1 - 2 \delta) {\left(1 + \frac{\delta}{2}\right)^{1/2} }}$, and for the iterates generated by \eqref{eq:WF_Updates} with a fixed step size of ${\mu}/{\| \z_0 \|^2} \leq 2/\beta$ we have that
$$
\mathrm{dist}^2(\z_k, \x) \leq  \epsilon^2 (1 - \frac{2 \mu}{\alpha})^k \| \x \|^2,
$$
with the best achievable convergence rate of
$$
\frac{2\mu}{\alpha} \leq \frac{h^2(\delta)}{c^2(\delta)} := \left(\frac{(1-\hat{\delta})(1-\epsilon)(2-\epsilon)}{(2+{\delta})(1+\epsilon)(2+\epsilon)} \right)^2,
$$
for the recovery of any $\x \in \mathbb{C}^N$, where $\hat{\delta} = \frac{\sqrt{2}(2+\epsilon)}{\sqrt{(1-\epsilon)(2-\epsilon})} \delta$.
\end{theorem}
\begin{remark}
For a deterministic model, \eqref{eq:WF_init} to hold for an unspecified, arbitrary $\x \in \mathbb{C}^N$ is equivalent to holding over all $\x \in \mathbb{C}^N$ \cite{yonel2020exact}. However in the probabilistic setting, the probability that \eqref{eq:WF_init} holds uniformly over all $\x \in \mathbb{C}^N$ is naturally more stringent than it is for an arbitrary $\x$. In Appendix \ref{sec:App}, we provide the uniformity of \eqref{eq:WF_init} for the Gaussian sampling model when $M = \mathcal{O}(N \log N)$.
\end{remark}

Complementing our performance guarantees in Theorem \ref{thm:Theorem1}, we can further derive an optimal step-size that achieves our upper bound on the rate of convergence in \eqref{eq:ab}. 
Knowing that $\beta = \beta' \| \x \|^2, \alpha = \alpha'/ \| \x \|^2$ with $\mu_{k} = \mu \leq 2/\beta'$ defined as the largest step-size allowed, and using \eqref{eq:RegCond_Final} we have that: 
\begin{equation}
\frac{2\left(h(\delta) - 1/\alpha'\right)}{c^2(\delta)} = \frac{r(\delta)\alpha'}{2},
\end{equation}
when $2 \mu/\alpha$ attains $r(\delta)$ with equality. 
Solving for $\alpha'$ yields $\alpha' = 2/h(\delta)$, which translates to a constant step-size of $\mu(\delta) = h(\delta)/c^2(\delta)$. 
Hence the WF iterations provably converge to a true solution at a constant step-size $\mu_k = \mu = \mathcal{O}(1)$ with $\mu > 0$ as long as $\delta \leq 0.184$ in \eqref{eq:WF_init}. 

The evolution of the optimal convergence rate $r(\delta)$ and the step-size $\mu(\delta)$ within our exact recovery regime is provided in Figure \ref{fig:3}.
The characterization of the convergence rate and the optimal step-size with respect to $\delta$ highlights a key contribution of our deterministic framework in Theorem \ref{thm:Theorem1}. 
Specifically considering the Gaussian model at a fixed number of samples at the proper complexity, picking a smaller $\delta$ is equivalent to relaxing the probability of success for exact recovery.
This in turn indicates moving up on the $\mu(\delta)$ curve such that the algorithm is ran with a larger step-size corresponding to the smaller $\delta$ value. 
Therefore in practical terms, the behavior in Figure \ref{fig:3} quantifies the trade-off between rate of convergence of the algorithm, and the probability of 
exactly recovering the ground truth $\x \in \mathbb{C}^N$. 
Alternatively an analogous phenomenon characterizes the trade-off between the number of samples and the algorithm performance, for a desired fixed probability of success. 
As a result, we establish an explicit relationship between the step-size chosen in practice, and the underlying model parameters such as the number of samples, and the probability of success for non-convex phase retrieval from intensity measurements via WF.

Note that the uniformity of \eqref{eq:WF_init} in the complex Gaussian sampling model was  shown
when $M = \mathcal{O}(N)$ in \cite{chen2017solving} via sample truncation, however not as a sufficient condition for exact recovery in the manner we established in Theorem \ref{thm:Theorem1}. 
Our framework essentially captures the impact of such a sample truncation scheme through its relation to the properties of the lifted forward model. 
In this sense, one could consider Theorem \ref{thm:Theorem1} as an abstraction of the  theoretical guarantees in \cite{chen2017solving}, where we have identified a novel, more minimal sufficient condition for arbitrary lifted forward models.
As a result of such an abstraction, we extend the favorable properties attained in \cite{chen2017solving} such as exact recovery with an $\mathcal{O}(1)$ step-size and the linear computational complexity to a more general problem setting which includes that of the original WF.
It can also be observed that the numerical cases evaluated in \cite{candes2015phase} for the Gaussian model are well within the range of values we identify for the validity of our Theorem \ref{thm:Theorem1}. 
Hence, the theoretical means developed for Theorem \ref{thm:Theorem1} are consistent with the convergence behavior demonstrated in \cite{candes2015phase}.

In establishing Theorem \ref{thm:Theorem1} for exact phase retrieval, we necessarily used the specific structure of the diagonal bias term in the expectation of the spectral matrix. 
This is in contrast to our work in \cite{bariscan2018}, in which the spectral matrix is an unbiased estimator of the lifted signal. 
Nonetheless, Corollary \ref{cor:Corr3} highlights a key advantage of the non-convex framework of WF.
Via the removal of convex relaxations and solving the perturbed problem in the lifted domain over the set of rank-1, PSD matrices, the RIP-type properties required by semi-definite programming and lifting-based approaches are relaxed to smaller, more specific domain of matrices. 
Under the lens of LRMR theory, WF not only offers computational advantages, but also less stringent theoretical means to achieve exact recovery if the step-size is properly controlled. 
The deterministic and less stringent nature of our recovery guarantees also opens up promising possibilities for the study of more structured models for problems such as wave-based imaging, where estimates on $\delta$ value would relate to parameters such as bandwidth, central frequency, or resolution \cite{Chai11, yonel2020exact}. 

Overall, we further stress a few notable outcomes of Theorem \ref{thm:Theorem1}, which include the following observations: 
\begin{itemize}
\item Via the established deterministic convergence framework given the concentration bound, our result proves that the restricted strong convexity property of the objective function is achieved with $\mathcal{O}(N \log N)$ samples for the complex Gaussian sampling model. 
This is a $\log N$ factor less than the sample complexity reported in \cite{sanghavi2017local}. 
\item Our sufficient condition has to hold only over the rank-1, PSD matrices, which is less stringent than those studied in the non-convex LRMR literature. 
Additionally, a universal upper bound on the relative distance-$\epsilon$ via Figure \ref{fig:1} is attained within our exact recovery regime. 
Hence, the concept of {sufficient accuracy} of the spectral initialization is captured by a quantitative measure. 
\item 
Another observation is that the upper bound on the concentration property of the spectral matrix in the phase retrieval problem requires a smaller constant than the one in the \emph{interferometric inversion problem} we studied in \cite{bariscan2018} ($0.184$ as opposed to $0.214$), in which the {relative phase information} of a pair of measurements is retained. 
This is indeed an intuitive outcome, as more information is lost when measurements are phaseless, compared to the interferometric case.  
A similar outcome is observed in the upper bound obtained for the geometric convergence rate of the algorithm, which approaches to $0$ for the case of interferometric inversion as $\delta \rightarrow 0$.  
As a result, the impact of the additional loss of phase information is directly captured in the sufficient conditions and the performance guarantees of the algorithm in solving the different types of quadratic systems of equations. 
\end{itemize}


\section{Robustness}\label{sec:Rob}
In this section, we assess the robustness of the WF algorithm in the presence of additive noise in the measurements.
We show that for the general problem setting of
\begin{equation}\label{eq:NoisyMeas}
\mathbf{y} = \mathcal{A}({\x \x^H}) + \bfeta,
\end{equation}
the results presented in Theorem \ref{thm:Theorem1} for the noise-free case in Section \ref{sec:Results} are attained upto a bounded perturbation for $\mathbb{E}[{\bfeta}] = \mathbf{0}$. {It should be noted that, the $\ell_2$ mismatch function minimized in the problem formulation fits the data for i.i.d. additive white Gaussian noise model $\{\eta_m\}_{m=1}^M$ in the maximum likelihood sense. 
Despite this, the results presented in this section have no specification on the distribution of the noise term $\bfeta$, similar to those of \cite{chen2017solving}, which were derived for the Poisson loss function.}

Our first goal is to establish the validity of the spectral initialization for our exact recovery guarantees with respect to the SNR of measurements in \eqref{eq:NoisyMeas} by utilizing the arguments developed over the lifted domain. 
Namely, for our subsequent convergence theory to hold, we derived numerical constraints on both the concentration bound (i.e., $\delta$), and the distance of the initial estimate obtained from the spectral method (i.e., $\epsilon$).  
These constraints characterize the amount of perturbation the algorithm can tolerate, which is stated in the following lemma and shown in Figure \ref{fig:4}. 

\begin{lemma}\label{lem:Lemma6}
Consider the spectral matrix formed by \eqref{eq:spectral}, using the noisy measurements in \eqref{eq:NoisyMeas}. Moreover, let the concentration bound in \eqref{eq:WF_init} hold as stated in the setup of Theorem \ref{thm:Theorem1}. Then, the spectral matrix $\mathbf{Y}$ satisfies,
\begin{equation}\label{eq:CBoundNoise}
\mathbb{E}_{\bfeta} \left[ \| \mathbf{Y} - ({\x} \x^H + \| \x \|^2 \mathbf{I}) \| \right] \leq \left( \delta + \frac{(2+\delta)}{\sqrt{\mathrm{SNR}}}\right) \| \x \|^2,
\end{equation}
where $\mathrm{SNR}$ stands for signal-to-noise-ratio, defined as $\mathrm{SNR} = {\| \mathcal{A}({\x \x^H}) \|^2}/\mathbb{E} [\| \bfeta \|^2]$. 
\end{lemma}
\begin{proof}
See Section \ref{sec:Prf6}.
\end{proof}

\begin{figure}
\centering
\includegraphics[scale=0.3]{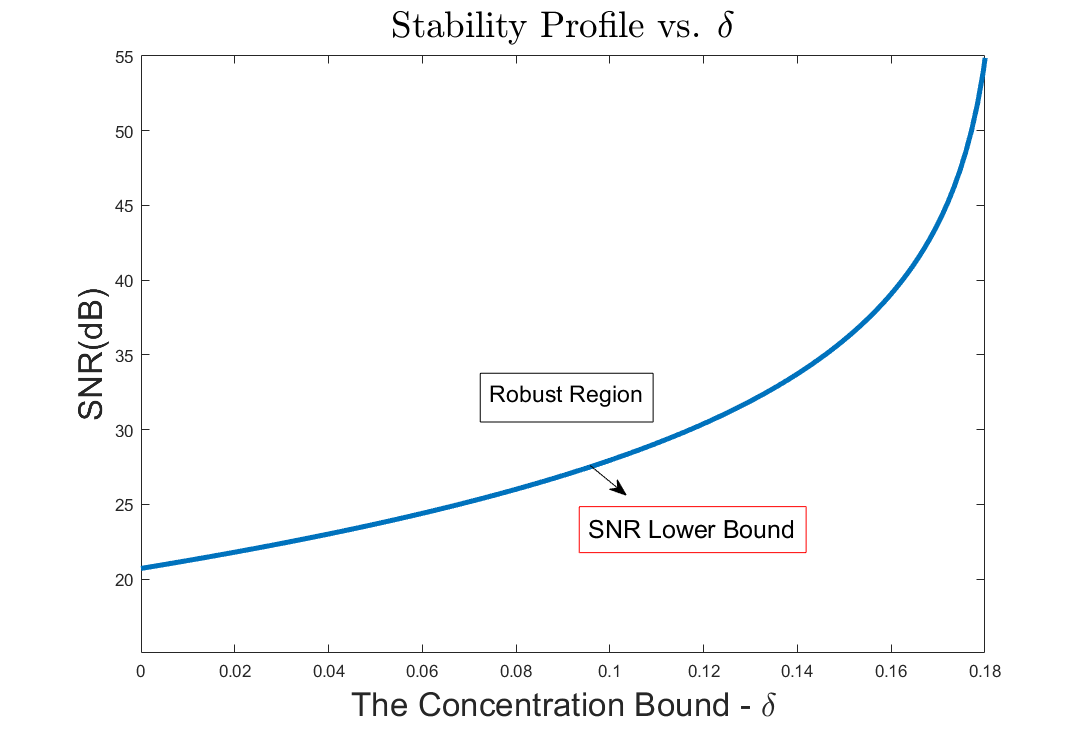}
\caption{\emph{The SNR lower bound required for our convergence framework to hold.} The truncation region refers to cases in which spectral initialization violates the sufficient conditions for Theorem \ref{thm:Thm2}. In this case, pre-processing techniques or other initialization schemes must be pursued.}
\label{fig:4}
\end{figure}

Analogous to the one-to-one relationship of the $\epsilon$-distance of the spectral initialization and the concentration bound $\delta$ in the noise-free phase retrieval problem, the noisy scenario has the additional dependence on the SNR of the measurements through $\tilde{\delta} := \delta + {(2+\delta)}/{\sqrt{\mathrm{SNR}}}.$
With the presence of the SNR term, there exists a level of noise as a function of $\delta$ beyond which the concentration bound in \eqref{eq:CBoundNoise} is insufficient to guarantee an effective spectral initialization. 
This restriction is directly determined by two constraints: $\epsilon < 1$ and $\hat{\delta} < 1$, in order to retain a valid regime where convergence arguments from Theorem \ref{thm:Theorem1} hold true for the noise-free component of the gradients. 
Thereby, we obtain a region over the ($\delta,\mathrm{SNR}$) domain such that the spectral method produces a valid estimator under $\tilde{\delta}$.  

A direct manner to determine this region is by enforcing $\tilde{\delta} \leq 0.184$, through which accuracy of spectral initialization and subsequent arguments within the $\mathit{E}(\epsilon)$ are preserved, yielding an SNR lower bound of  
\begin{equation}\label{eq:Sec3_5_eq2}
\mathrm{SNR(dB)} \geq 20 \log \frac{2+\delta}{0.184 - \delta}.
\end{equation}
Figure \ref{fig:4} depicts that the lowest SNR value of $20.7 \mathrm{dB}$ is attained at $\delta = 0$. 
Although it is derived in a straight-forward manner, \eqref{eq:Sec3_5_eq2} is the best lower bound that can be characterized by our framework. 
This is rather surprising at first glance, since the $\hat{\delta}$ term is only affected by the perturbation resulting from noise through the $\epsilon$ parameter, as $\delta$ and $\hat{\delta}$ are properties of the underlying lifted forward model $\mathcal{A}$, which is independent of noise. 
However, these still prove to be consequential for the stability of the algorithm under additive noise because of constraints that arise from the convergence arguments, beyond those related to the validity of the spectral initialization. 

In particular, under additive noise and the assumptions of Lemma \ref{lem:Lemma6}, within the valid region for the spectral initialization defined by \eqref{eq:Sec3_5_eq2}, the convergence guarantees of WF are perturbed by a constant factor that is a function of SNR. 
\begin{theorem}\label{thm:Thm2}
Assume that the assumptions of Lemma \ref{lem:Lemma6} hold. Then, for the identical procedure and values of constants $\alpha, \beta$ in the setup of Theorem \ref{thm:Theorem1}, we have
\begin{equation}\label{eq:thm2eq}
\mathbb{E}_{\bfeta} \left[\mathrm{dist}(\z_k, \x) \right] \lessapprox  \epsilon (1 - \frac{2 \mu}{\alpha})^{\frac{k}{2}} \| \x \| + {\alpha'} \frac{(2+\delta)}{\sqrt{\mathrm{SNR}}} \| \x \|,
\end{equation}
where $\mu = \mu_k/ \| \z_0 \|^2 \leq 2/\beta$, $\alpha' = \mathcal{O}(1)$ such that $\alpha = \alpha' / \| \z_0 \|^2$, and $\epsilon^2 \leq 1 + \sqrt{1 + \frac{\tilde{\delta}}{2}} - 2 \sqrt{(1 - 2 \tilde{\delta}) {\left(1 + \frac{\tilde{\delta}}{2}\right)^{1/2} }}$, with $\tilde{\delta} = \delta + \frac{(2+\delta)}{\sqrt{\mathrm{SNR}}}$. 
\end{theorem}
\begin{proof}
See Section \ref{sec:Prf7}.
\end{proof}

As a result of Theorem \ref{thm:Thm2}, 
we observe a crucial element for determining the trade-off between the $\alpha$ and $\beta$ parameters.  
The stability guarantees directly incentivize allocating a small value for the parameter $\alpha$, which is inversely related to the magnitude of the $\beta$ parameter. 
Since $\mu \leq {2}/{\beta}$ by definition, a tighter stability bound requires a trade-off from the step size of the algorithm. 
This outcome is indeed expected, as lower SNR in the measurements means more variance for the gradient estimates at the iterative refinement stage, hence one should take less confident steps to counter inaccurate update terms. 
Our framework perfectly captures this phenomenon, and requires small step sizes for improved stability in the algorithm performance while optimizing the noisy landscape of the objective function over the signal domain. 


Furthermore, to guarantee that the iterates formed using the noisy measurements remain in the $\mathit{E}(\epsilon)$, there is an effective upper bound on the $\alpha$ parameter such that
\begin{equation}\label{eq:alphbd}
\alpha \| \x \|^2 \leq \frac{\epsilon \sqrt{\mathrm{SNR}}}{(1 + \epsilon) (2 + \delta)}. 
\end{equation} 
Equivalently, this is a lower bound on the value of $1/(\alpha \| \x \|^2)$ in \eqref{eq:RegCond_Final}, which requires 
\begin{equation}\label{eq:alphbd_2}
\frac{(1 + \epsilon) (2 + \delta)}{\epsilon \sqrt{\mathrm{SNR}}} \leq \frac{1}{\alpha \| \x \|^2} < (1-\hat{\delta})(1-\epsilon)(2-\epsilon)
\end{equation}  
to be satisfied for a finite $\beta$ to exist to attain a practical step-size for the algorithm. 
Thereby, on expectation, iterative updates that are contractions with respect to the distance metric can be achieved, if $\mathrm{SNR}$ is sufficiently high to satisfy
\begin{equation}
\frac{(1 + \epsilon) (2 + \delta)}{\epsilon \sqrt{\mathrm{SNR}}}  < (1-\hat{\delta})(1-\epsilon)(2-\epsilon),
\end{equation}
where both $\epsilon$ and $\hat{\delta}$ are $\mathrm{SNR}$ dependent. 
Since the left- and right-hand-sides of the inequality monotonically decrease and increase, respectively, with increasing $\mathrm{SNR}$, there exists a transition point at any fixed $\delta < 0.184$ value beyond which the inequality holds true. 
This numerical characterization of the $\mathrm{SNR}$ requirements for the convergence of WF precisely corresponds to the lower bound in \eqref{eq:Sec3_5_eq2}. 

Beyond our characterization of the $(\delta, \mathrm{SNR})$ domain, pre-processing techniques developed in the phase retrieval literature \cite{Zhang2017a, mondelli2017fundamental, gao2017phaseless} can be adopted to improve the SNR bound by relaxing our requirements on $\delta$ through the value of $\epsilon$. 
As a result, alternative techniques that improve the accuracy of the spectral estimation can be pursued for practical purposes in order to sustain the convergence guarantees given our sufficient condition below the provided SNR lower bound.

\section{Proofs}\label{sec:Apdx}
In this section we present the proofs of our arguments used in establishing Theorem \ref{thm:Theorem1}, and its corollaries. 
Notably, our results are derived in a deterministic framework based on geometric arguments unlike the probabilistic theory of original WF theory \cite{candes2015phase}. 
The probabilistic nature of the convergence analysis prominent in phase retrieval literature is thereby compressed into a single condition on the lifted forward model. 
From this point on, we frequently use $\tilde{\V} = \v \v^H$ to denote the lifted signal corresponding to an element $\v \in \mathbb{C}^N$ for notational brevity. 
\subsection{Proofs of Lemma \ref{lem:Lemma1} and Corollary \ref{cor:Corr3}}\label{sec:Prf1}
\subsubsection{Lemma \ref{lem:Lemma1}}
Reprising \eqref{eq:WF_init}, from the definition of the lifted forward model $\mathcal{A}$, and spectral matrix $\mathbf{Y}$ we have for any $\x \in \mathbb{C}^N$:
\begin{equation}\label{eq:Eqn1}
\| \frac{1}{M} \mathcal{A}^H \mathcal{A} (\tilde{\X}) - ( \tilde{\X} + \| \x \|^2 \mathbf{I} ) \| \leq \delta \| \x \|^2.
\end{equation}
Over the set of rank-1 matrices matrices, i.e., $\mathit{R}_1 = \{ \mathbf{u}\mathbf{v}^H: \mathbf{u}, \mathbf{v} \in \mathbb{C}^N \}$, we define the operator $\mathcal{R}: \mathit{R}_1 \rightarrow \mathrm{span}(\mathbf{I}) \subset \mathbb{C}^{N \times N}$, such that $\mathcal{R}(\mathbf{u}\mathbf{v}^H) = (\mathbf{v}^H \mathbf{u}) \mathbf{I}$. Then, we define $(\frac{1}{M} \mathcal{A}^H \mathcal{A} - \mathcal{I} - \mathcal{R}) (\tilde{\X}) = {\Delta}(\tilde{\X})$, and by \eqref{eq:Eqn1} we have that $\| {\Delta}(\tilde{\X}) \| \leq \delta \| \x \|^2$. 

Next, we represent the rank-1, PSD matrix $\tilde{\Z} = \z \z^H$ as a linear combination of rank-1 elements in $\mathit{R}_1$. Letting $\mathbf{e} = \mathbf{z} - \mathbf{x}$, we have
\begin{equation}\label{eq:Eqn2}
\tilde{\mathbf{E}} = \tilde{\Z} - \tilde{\X} - \mathbf{e} \x^H - \x \mathbf{e}^H,
\end{equation}
hence $\tilde{\Z} = \tilde{\mathbf{E}} +  \mathbf{e} \x^H + \x \mathbf{e}^H + \tilde{\X}$. In the range of $\mathcal{R}$, distributing over the terms on the right hand side we obtain
\begin{equation}
\mathcal{R}(\tilde{\Z}) = (\| \mathbf{e}  \|^2 + \x^H \mathbf{e} + \mathbf{e}^H \x + \| \x \|^2) \mathbf{I},
\end{equation}
which, from the definition of $\mathbf{e}$, precisely equals to $\| \z \|^2 \mathbf{I}$. 
Thus having $\mathit{R}^{+}_1 \subset \mathit{R}_1$, by moving $\x$ to the left-hand-side $\mathcal{R}$ satisfies,
\begin{equation}
\mathcal{R}(\tilde{\Z} - \tilde{\X} ) = \| \z \|^2 - \| \x \|^2.
\end{equation}
 
Since $\mathcal{I} + \mathcal{R}$ is linear over elements in $\mathit{R}^{+}_1$, ${\Delta}$ is necessarily linear as well knowing that $\mathcal{A}^H \mathcal{A}$ is linear over all $\mathbb{C}^{N \times N}$. 
In addition, since $\langle \mathcal{R}(\tilde{\X}) , \tilde{\X} \rangle_F = \| \tilde{\X} \|_F^2$, $\mathcal{R}$ is self-adjoint, hence ${\Delta}$ is a Hermitian operator by definition. Finally, since $\| {\Delta} ( \tilde{\X} ) \| \leq \delta \| \x \|^2$, we have
\begin{equation}\label{eq:Eqn2a}
\underset{\v \in \mathbb{C}^N / \{0 \}}{\mathrm{max}} \  \frac{| \mathbf{v}^H {\Delta} ( \tilde{\X} ) \mathbf{v} |}{\| \v \|^2} \leq \delta \| \x \|^2. 
\end{equation}
Equivalently, utilizing the view in the lifted domain, we have, $| \mathbf{v}^H {\Delta} ( \tilde{\X} ) \mathbf{v} | = | \langle {\Delta} ( \tilde{\X} ), \tilde{\V}\rangle_F | = | \langle {\Delta} ( \tilde{\V} ), \tilde{\X}\rangle_F |$. 
Hence the concentration property directly implies uniformity of the bound over the set of the domain of operator $\Delta$, i.e., the set of rank-1, PSD matrices as
\begin{equation}\label{eq:Eqn2b}
\underset{\v \in \mathbb{C}^N / \{0 \}}{\mathrm{max}} \  \frac{| \mathbf{x}^H {\Delta} ( \v \v^H ) \mathbf{x} |}{\| \x \|^2 \| \v \|^2} \leq \delta,
\end{equation}
again, for any fixed $\x \in \mathbb{C}^N$.	For an arbitrary deterministic map that already satisfies \eqref{eq:Eqn1}, this completes the proof as the condition also holds for $\hat{\v}$ that maximizes $| \langle {\Delta} ( \tilde{\V} ), \tilde{\X}\rangle_F |/ \| \x \|^2$ as $| \langle {\Delta} ( \tilde{\V} ), \hat{\v} \hat{\v}^H \rangle_F |/ \| \hat{\v} \|^2$, which, by definition, is the spectral norm of $ \| \Delta (\v \v^H) \|$. For uniformity in the complex Gaussian model, we refer the reader to Appendix \ref{sec:App}.


\subsubsection{Corollary \ref{cor:Corr3}}
The proof of Corollary \ref{cor:Corr3} then directly follows from \eqref{eq:WF_init} and 
the definition of the spectral matrix. 
From \eqref{eq:Eqn1} and the definition of the spectral norm, we have
\begin{equation}
\frac{| \v^H \left( \frac{1}{M} \mathcal{A}^H \mathcal{A} (\tilde{\X}) - ( \tilde{\X} + \| \x \|^2 \mathbf{I} ) \right) \v |}{\| \v \|^2} \leq  \delta \| \x \|^2,
\end{equation}
for any $\v \in \mathbb{C}^N / \{0\}$. Hence, for $\v = \x$, using the representation in the lifted domain via the definition of the Frobenius inner product, we have
\begin{equation}
\begin{split}
{| \langle \frac{1}{M} \mathcal{A}^H \mathcal{A} (\tilde{\X}) - ( \tilde{\X} + \| \x \|^2 \mathbf{I} ), \x \x^H \rangle_F |} &\leq \delta {\| \x \|^4}, \\
|\frac{1}{M} \| \mathcal{A}(\tilde{\X}) \|^2 - \| \tilde{\X} \|^2_F - \langle \| \x \|^2 \mathbf{I} , \tilde{\X} \rangle_F | & \leq \delta {\| \tilde{\X} \|^2_F}.
\end{split}
\end{equation}
From the definition of the Frobenius inner product with the identity matrix $\mathbf{I}$ in $\mathbb{C}^{N \times N}$, only the diagonal elements are multiplied and summed with $\tilde{\X} = \x \x^H$, which corresponds to $\| \x \|^2$. Hence, we get, for any $\x \in \mathbb{C}^N$,
\begin{equation}
|\frac{1}{M} \| \mathcal{A}(\tilde{\X}) \|^2 - 2\| \tilde{\X} \|^2_F | \leq \delta {\| \tilde{\X} \|^2_F}.
\end{equation}
Thus, the condition in Corollary \ref{cor:Corr3} is directly implied by \eqref{eq:Eqn1}. 
Furthermore, again using the fact that $\Delta$ is Hermitian, consider the maxima over the unit sphere $\| \x \| = 1$ as
\begin{align}
\underset{ \| \x \| = 1}{\text{max}} \| \Delta (\x \x^H ) \| &= \underset{ \| \x \| = 1}{\text{max}}  \underset{ \| \v \| = 1}{\text{max}} |  \v^H \Delta ( \x \x^H )  \v | \\
&= \underset{ \| \x \| = 1}{\text{max}}  \underset{ \| \v \| = 1}{\text{max}} | \langle \sqrt{\Delta} ( \x \x^H ) , \sqrt{\Delta} (\v \v^H) \rangle_F  |,
\end{align}
where $\sqrt{\Delta}: \mathbb{C}^{N \times N} \rightarrow \mathbb{C}^{N \times N}$ is a (non-unique) square root of the operator $\Delta$. 
Then from Cauchy-Schwartz we have that
\begin{equation}
\underset{ \| \x \| = 1}{\text{max}} \| \Delta (\x \x^H ) \| \leq \underset{ \| \x \| = 1}{\text{max}}  \underset{ \| \v \| = 1}{\text{max}} \| \sqrt{\Delta} ( \x \x^H ) \|_F \| \sqrt{\Delta} ( \v \v^H ) \|_F
\end{equation}
in which we clearly have a maximization that is fully split such that
\begin{equation}\label{eq:e4}
\underset{ \| \x \| = 1}{\text{max}} \| \Delta (\x \x^H ) \| \leq \underset{ \| \x \| = 1}{\text{max}}  \| \sqrt{\Delta} ( \x \x^H ) \|_F^2.  
\end{equation}
From our construction of $\sqrt{\Delta}$, the upper bound in \eqref{eq:e4} yields
\begin{equation}\label{eq:e5}
\underset{ \| \x \| = 1}{\text{max}} \| \Delta (\x \x^H ) \|  \leq \underset{ \| \x \| = 1}{\text{max}} | \x^H \Delta ( \x \x^H) \x |.
\end{equation} 
Since $\| \Delta (\x \x^H ) \| = \underset{\| \v \| = 1}{\text{max}} \ | \v^H \Delta ( \x \x^H )  \v | \geq | \x^H \Delta ( \x \x^H )  \x |$ by definition, \eqref{eq:e5} implies that the maxima necessarily occurs at $\v = \x / \| \x \|$, as
\begin{equation}
\underset{ \| \x \| = 1}{\text{max}} \| \Delta (\x \x^H ) \|  =  \underset{ \| \x \| = 1}{\text{max}} | \langle \Delta ( \x \x^H) , \x\x^H \rangle_F |. 
\end{equation}
Hence, the concentration bound holds uniformly with a constant $\delta$ that is equivalent to the restricted isometry constant of the lifted forward model over the set of rank-1, PSD matrices.

\subsection{Proof of Lemma \ref{lem:Lemma2}} \label{sec:Prf2}
As shown in \cite{candes2015phase} we know that \eqref{eq:WF_init} implies $| \v_0^H \x |^2 \geq (1 - 2\delta) \| \x \|^2$, where $\| \v_0 \| = 1$ is the leading eigenvector of the spectral matrix $\mathbf{Y}$. 
Using Corollary \ref{cor:Corr3}, we know that $\lambda_0 = \| \y \|/ \sqrt{2M}$ is an estimate for the energy of the unknown signal $\x$ such that $(\sqrt{1 - \delta/2}) \| \x \|^2 \leq \lambda_0 \leq (\sqrt{1 + \delta/2}) \| \x \|^2$. Using the definition of the distance metric, and the lower bound from \cite{candes2015phase}, for $\z_0 = \sqrt{\lambda_0} \v_0$, we have
\begin{equation}
\mathrm{dist}^2 (\z_0, \x ) \leq \left(\frac{\lambda_0}{\| \x \|^2} + 1 - 2 \sqrt{\frac{\lambda_0}{\| \x \|^2}} \| \x \| \sqrt{1 - 2 \delta} \right) \| \x \|^2.
\end{equation}
Since the right-hand-side is a convex quadratic function of $\sqrt{\lambda_0}$, its maximum value is reached at the boundary values of $\sqrt{\lambda_0}$. Setting $\lambda_0 = (\sqrt{1 + \delta/2})$, the upper bound is monotonically greater than at $\lambda_0 = (\sqrt{1 - \delta/2})$ for all valid values for $\delta$, hence we conclude that
\begin{equation}
\mathrm{dist}^2 (\z_0, \x ) \leq \left(\sqrt{1 + \frac{\delta}{2}} + 1 - 2 \sqrt{(1 - 2 \delta)\sqrt{1 + \frac{\delta}{2}} } \right) \| \x \|^2.
\end{equation}

\subsection{Proof of Lemma \ref{lem:Lemma3}} \label{sec:Prf3}
\subsubsection{Proof of the Upper Bound}
Let $\hat{\x}$ be the closest solution in $\mathit{P}$ to an arbitrary $\z \in E(\epsilon)$. 
From reverse triangle inequality we have $(1-\epsilon) \| \x \| \leq \| \z \| \leq (1+\epsilon) \| \x \|$. Setting $\mathbf{e} = \z - \hat{\x}$, by \eqref{eq:Eqn2} we have
\begin{equation}\label{eq:EQN}
\tilde{\Z} - \tilde{\X} = \tilde{\mathbf{E}} + \mathbf{e} \hat{\x}^H + \hat{\x} \mathbf{e}^H. 
\end{equation}
Then, for the Frobenius norm of the error in the lifted domain, we have
\begin{equation}
\| \tilde{\Z} - \tilde{\X} \|_F \leq \| \tilde{\mathbf{E}} \|_F + \| \mathbf{e} \hat{\x}^H \|_F + \| \hat{\x} \mathbf{e}^H \|_F. 
\end{equation}
Since all the elements on the right-hand-side are rank-1, and $\z \in \mathit{E}(\epsilon)$, by definition, we have $\| \cdot \| = \| \cdot \|_F$, and
\begin{equation}
\| \tilde{\Z} - \tilde{\X} \|_F \leq \| \mathbf{e} \|^2 + 2 \| \x \| \| \mathbf{e} \| \leq (2+\epsilon) \| \mathbf{e} \| \| \x \|,
\end{equation}
which yields the upper bound as $\| \mathbf{e} \| = \mathrm{dist}(\z, \x)$. 

\subsubsection{Proof of the Lower Bound in Lemma \ref{lem:Lemma3}}
Expanding $\| \tilde{\Z} - \tilde{\X} \|_F$ by the definition of the Frobenius inner product, we have
\begin{equation}\label{eq:Eqn3}
\| \tilde{\Z} - \tilde{\X} \|_F^2 = \| \tilde{\Z} \|_F^2 + \| \tilde{\X} \|_F^2 - 2 \text{Re} \langle \tilde{\Z}, \tilde{\X} \rangle_F.
\end{equation}
Due to rank-1 property, the Frobenius inner product reduces to $2 \text{Re} \langle \tilde{\Z}, \tilde{\X} \rangle_F = 2 | \langle \z, \x \rangle |^2$, and $\| \tilde{\Z} \|_F^2 = \| \z \|^4$, $\| \tilde{\X} \|_F^2 = \| \x \|^4$. 
Hence, \eqref{eq:Eqn3} equals to $\| \z \|^4 + \| \x \|^4 - 2 | \langle \z , \x \rangle |^2$, and
\begin{equation}
\begin{split}
(\| \z \|^4 -  | \langle \z , \x \rangle |^2 ) &+  (\| \x \|^4 - | \langle \z , \x \rangle |^2 ) = \\
&(\| \z \|^2 +  | \langle \z , \x \rangle |)(\| \z \|^2 -  | \langle \z , \x \rangle |) + \\
& (\| \x \|^2 + | \langle \z , \x \rangle |)(\| \x \|^2 - | \langle \z , \x \rangle |). 
\end{split}
\label{eq:Eqn4}
\end{equation}
Since $\text{dist}^2(\z, \x) = \| \z \|^2 + \| \x \|^2 - 2 | \langle \z, \x \rangle | = \| \z - \hat{\x} \|^2 \geq 0$, we can lower bound \eqref{eq:Eqn3} using \eqref{eq:Eqn4} as
\begin{equation}\label{eq:Eqn5}
\begin{split}
\|  \tilde{\z} - \tilde{\x}  \|_F^2 \geq \text{min} &\left( (\| \z \|^2 +  | \langle \z , \x \rangle |), (\| \x \|^2 + | \langle \z , \x \rangle |)\right) \\
&\times \left( \| \z \|^2 + \| \x \|^2 - 2 | \langle \z, \x \rangle | \right).
\end{split}
\end{equation}
Knowing that $\text{dist}^2(\z, \x) \leq \epsilon^2 \| \x \|^2$ and the result from the reverse triangle inequality on $\| \z \|$, the terms within the minimization are further lower bounded using
\begin{eqnarray}
 2 | \langle \z, \x \rangle |  & \geq & \| \z \|^2 + \| \x \|^2 - \epsilon^2 \| \x \|^2 \nonumber \\
 | \langle \z, \x \rangle |  & \geq & (1 - \epsilon) \| \x \|^2. 
\end{eqnarray}
We then get the bound on the scalar multiplying $\text{dist}^2(\z, \x)$ as
\begin{equation}
\begin{split}
\text{min} \left( (\| \z \|^2 +  | \langle \z , \x \rangle |), (\| \x \|^2 + | \langle \z , \x \rangle |)\right)  \\
\geq  ((1 - \epsilon)^2 + (1-\epsilon)) \| \x \|^2,
\end{split}
\end{equation}
which yields the lower bound of Lemma \ref{lem:Lemma3}
\begin{equation}
 \| \tilde{\z} - \tilde{\x} \|_F \geq \sqrt{(1-\epsilon)(2-\epsilon)} \ \text{dist} (\z , {\x} ) \| \x \|,
\end{equation}
and the proof is complete. 
\subsection{Proof of Lemma \ref{lem:Lemma4}} \label{sec:Prf4}
From Lemma \ref{lem:Lemma1}, we have 
\begin{equation}\label{eq:Eqn6}
\begin{split}
&\frac{1}{M} \| \mathcal{A}(\tilde{\Z} - \tilde{\X}) \|^2 = \langle \mathcal{A}^H \mathcal{A} (\tilde{\Z} - \tilde{\X}) \rangle_F = \\
&\langle \tilde{\Z} - \tilde{\X} + (\| \z \|^2 - \| \x \|^2) \mathbf{I} + {\Delta} (\tilde{\Z} - \tilde{\X} ) , \tilde{\Z} - \tilde{\X} \rangle_F.
\end{split}
\end{equation}  
From the linearity of the inner product, \eqref{eq:Eqn6} becomes $\| \tilde{\Z} - \tilde{\X} \|_F^2 + (\| \z \|^2 - \| \x \|^2) (\langle \mathbf{I}, \tilde{\z} \rangle_F - \langle \mathbf{I}, \tilde{\x} \rangle_F) + \langle {\Delta} (\tilde{\Z} - \tilde{\X} ) , \tilde{\Z} - \tilde{\X} \rangle_F$. 
Since the Frobenius inner product of a matrix with the identity matrix $\mathbf{I}$ is simply the sum of its diagonal terms, and the lifted signals have the auto-correlation of their entries at diagonals, the second term reduces to $(\| \z \|^2 - \| \x \|^2)^2 = \| \z \|^4 + \| \x \|^4 - 2 \| \z \|^2 \| \x \|^2$. From Cauchy-Schwartz, $(\| \z \|^2 - \| \x \|^2)^2 \leq \| \z \|^4 + \| \x \|^4 - 2| \langle \z, \x \rangle |^2 = \| \tilde{\Z} - \tilde{\X} \|_F^2$, by definition. Hence, we obtain
\begin{equation}\label{eq:Bound1}
\frac{1}{M} \| \mathcal{A}(\tilde{\Z} - \tilde{\X}) \|^2 \leq 2 \| \tilde{\Z} - \tilde{\X} \|_F^2 +  \langle {\Delta} (\tilde{\Z} - \tilde{\X} ) , \tilde{\Z} - \tilde{\X} \rangle_F,
\end{equation}
and in the other direction, since $(\| \z \|^2 - \| \x \|^2)^2$ is the square of a real valued quantity, it is lower bounded by $0$, which yields
\begin{equation}\label{eq:Bound2}
\frac{1}{M}\| \mathcal{A}(\tilde{\Z} - \tilde{\X}) \|^2 \geq \| \tilde{\Z} - \tilde{\X} \|_F^2 +  \langle {\Delta} (\tilde{\Z} - \tilde{\X} ) , \tilde{\Z} - \tilde{\X} \rangle_F.
\end{equation}
It remains to upper bound the quantity $| \langle {\Delta} (\tilde{\Z} - \tilde{\X} ) , \tilde{\Z} - \tilde{\X} \rangle_F |$. 
Using the definition in \eqref{eq:EQN}, and the linearity of ${\Delta}$ 
, from Cauchy Schwartz inequality, we have, for any $\z \in \mathit{E}(\epsilon)$,
\begin{equation}\label{eq:Eqn7}
\begin{split}
& | \langle {\Delta} (\tilde{\E} + \e \hat{\x}^H + \hat{\x} \e^H ) , \tilde{\Z} - \tilde{\X} \rangle_F | \\
&\leq \| \tilde{\Z} - \tilde{\X} \|_* \left(\| {\Delta}(\tilde{\E}) \|  + \| {\Delta}({\e}\hat{\x}^H) \| + \| {\Delta}(\hat{\x}{\e}^H) \| \right), \\
&\leq  \sqrt{2} \| \tilde{\Z} - \tilde{\X} \|_F \delta (\| \tilde{\E} \| + \| {\e}\hat{\x}^H \|  + \| \hat{\x}{\e}^H \|)\\
&=   \sqrt{2} \| \tilde{\Z} - \tilde{\X} \|_F \delta ( \| \e \|^2 + 2 \| \x \| \| \e \| )  \\
&\leq \sqrt{2} (2 + \epsilon) \delta \| \tilde{\Z} - \tilde{\X} \|_F \ \mathrm{dist}(\z, \x) \| \x \|. 
\end{split}
\end{equation}
Finally, using the lower bound from Lemma \ref{lem:Lemma3} where $\epsilon < 1$, we obtain
\begin{equation}\label{eq:Eqn8}
\sqrt{2} (2 + \epsilon) \delta \| \tilde{\Z} - \tilde{\X} \|_F \ \mathrm{dist}(\z, \x) \| \x \| \leq \frac{\sqrt{2} (2 + \epsilon)\delta }{\sqrt{(1-\epsilon)(2-\epsilon)}} \| \tilde{\Z} - \tilde{\X} \|_F^2.
\end{equation}
Thereby, setting $\hat{\delta}  = \frac{\sqrt{2} (2 + \epsilon)}{\sqrt{(1-\epsilon)(2-\epsilon)}} \delta$, and combining the bounds \eqref{eq:Bound1}, \eqref{eq:Bound2} and \eqref{eq:Eqn8}, the proof is complete. 

\subsection{Proof of Lemma \ref{lem:Lemma5}} \label{sec:Prf5} 

Recall the definition of the gradient in \eqref{eq:GradEq}. 
By Lemma \ref{lem:Lemma4} a RIP-type property is satisfied locally for ${\z \z^H} - {\x \x^H}$ if $\z \in E(\epsilon)$. As a result we can express $\nabla f(\z) = \tilde{\Z} \z - \tilde{\X}\z + (\| \z \|^2 - \| \x \|^2)\z + {\Delta}(\tilde{\Z} - \tilde{\X}) \z = \| \z \|^2 \z - (\hat{\x}^H \z)  \hat{\x} + (\| \z \|^2 - \| \x \|^2)\z + {\Delta}(\tilde{\Z} - \tilde{\X}) \z$, with $\hat{\x}$ again denoting the closest solution in $\mathit{P}$ to a given $\z \in \mathbb{C}^N$. Then, we upper bound $\nabla f (\z)$ as follows:
\begin{equation}
\begin{split}
\| \nabla f(\z) \| \leq \| \| \z \|^2 \z - (\hat{\x}^H \z)  \hat{\x} \| &+ | \| \z \|^2 - \| \x \|^2|  \| \z \|
 \\
&+ \| {\Delta}(\tilde{\Z} - \tilde{\X}) \| \| \z \|  \\
\end{split}
\end{equation}
from which, knowing that $\| \z \|^2 \z - (\hat{\x}^H \z)  \hat{\x} = (\| \z \|^2 - \hat{\x}^H \z) \z + (\hat{\x}^H \z)(\z - \hat{\x})$, where $\e = \z - \hat{\x}$ with $\hat{\x}$ again denoting the closest solution to $\z$ in $\mathit{P}$, and $| \| \z \|^2 - \| \x \|^2| \leq \| \tilde{\Z} - \tilde{\X} \|_F$, we have
\begin{equation}
\begin{split}
\| \nabla f(\z) \| \leq \| \z \| \bigl(&\mathrm{dist}(\z, \x)  (\| \z \|  + \| \x \| ) \\ 
& + \| \tilde{\Z} - \tilde{\X} \|_F +  \| {\Delta}(\tilde{\Z} - \tilde{\X}) \| \bigr)  .
\end{split}
\end{equation}
Again considering the expression $\tilde{\Z} - \tilde{\X} = \tilde{\E} + \e \hat{\x}^H + \hat{\x} \e^H$, we have $\| {\Delta}(\tilde{\Z} - \tilde{\X}) \| \leq \| \Delta ( \tilde{\E} ) \| + \| \Delta (\e \hat{\x}^H + \hat{\x} \e^H )\| $. 
Since $\e \hat{\x}^H + \hat{\x} \e^H$ is at most of rank-2 by definition, let $\e \hat{\x}^H + \hat{\x} \e^H = \sum_{i = 1}^2 \lambda_i \v_i \v_i^H$, by which we obtain  $ \| \Delta (\e \hat{\x}^H + \hat{\x} \e^H )\| \leq | \lambda_1 | \| \Delta (\v_1 \v_1^H) \| + |\lambda_2| \| \Delta (\v_2 \v_2^H ) \|$. 
Thereby, using Lemma \ref{lem:Lemma1}, we have $\| {\Delta}(\tilde{\Z} - \tilde{\X}) \| \leq \delta ( \| \e \|^2 + | \lambda_1 | + | \lambda_2 |)$. Furthermore, $| \lambda_1 | + | \lambda_2 |$ precisely corresponds to the nuclear norm of $\e \hat{\x}^H + \hat{\x} \e^H$, which is upper bounded as $\| \e \hat{\x}^H + \hat{\x} \e^H \|_* \leq \| \e \hat{\x}^H \|_* + \| \hat{\x} \e^H \|_* \leq 2 \| \e \| \| \x \|$. As a result, invoking the upper bound on the error in the lifted domain via Lemma \ref{lem:Lemma3}, we obtain
\begin{equation}\label{eq:eqL5}
\begin{split}
\| \nabla f(\z) \|  &\leq \| \z \| \left( (\| \z \| + \| \x \| )+ (1 + \delta) (2 + \epsilon) \| \x \| \right) \mathrm{dist}(\z, \x) \\
& \leq (1 + \epsilon) (2+ \epsilon) (2+ {\delta}) \| \x \|^2 \mathrm{dist}(\z, \x). 
\end{split}
\end{equation}
In \eqref{eq:eqL5} we've used the fact that, for $\z \in \mathit{E}(\epsilon)$, $\| \z \| \leq (1 + \epsilon)\| \x \|$. Hence, the proof of local Lipschitz differentiability is complete, with a constant $c (\delta) = (1 + \epsilon) (2+ \epsilon) (2+ {\delta})$. 

\subsection{Proof of Lemma \ref{lem:Lemma6}} \label{sec:Prf6}
Using the definition of the spectral matrix as $\mathbf{Y} = \frac{1}{M} \mathcal{A}^H (\mathbf{y})$, we have, for any $\x \in \mathbb{C}^N$, and any realization of $\bfeta \in \mathbb{R}^M$
\begin{align}
\| \mathbf{Y} - (\tilde{\X} + \| \x \|^2 \mathbf{I}) \| \leq \| \frac{1}{M} \mathcal{A}^H \mathcal{A} (\tilde{\X} ) &- (\tilde{\X} + \| \x \|^2 \mathbf{I}) \\
&+ \| \frac{1}{M} \mathcal{A}^H(\bfeta) \|, \nonumber
\end{align}
where $\mathbf{y} = \mathcal{A}(\tilde{\X}) + \bfeta$, and the inequality simply follows from the triangle inequality.
Under the assumption that our sufficient condition in \eqref{eq:WF_init} holds, we get
\begin{equation}\label{eq:Sec3_5proofLem1}
\| \mathbf{Y} - (\tilde{\X} + \| \x \|^2 \mathbf{I}) \|  \leq \delta \| \x \|^2 + \| \frac{1}{M} \mathcal{A}^H(\bfeta) \|.
\end{equation}
Then, using the definition of the spectral norm, we have
\begin{equation}\label{eq:eqL6}
\begin{split}
 \| \frac{1}{M} \mathcal{A}^H(\bfeta) \| 
 &= \underset{\z \in \mathbb{C}^N, \| \z \| = 1}{\mathrm{max}} \ | \langle \frac{1}{M} \mathcal{A}^H(\bfeta), \z \z^H \rangle_F | \\
 &= \underset{\z \in \mathbb{C}^N, \| \z \| = 1}{\mathrm{max}} \ | \langle \frac{1}{M}\bfeta, \mathcal{A}(\z \z^H) \rangle | \\
 & \leq \frac{1}{M} \| \bfeta \| \| \mathcal{A}(\z \z^H) \| \leq \frac{\sqrt{2 + \delta}}{\sqrt{M}} \| \bfeta \|
\end{split}
\end{equation}
where we used the adjoint definition over the Frobenius inner product for $\mathcal{A}$, followed by Cauchy-Schwartz inequality, and the upper bound from Corollary \ref{cor:Corr3} $ \| \mathcal{A} ({\z \z^H}) \|^2/M \leq ({2 + \delta}) \| {\z \z^H} \|^2_F$, with $\| \z \| = 1$. 
Reorganizing the last inequality in \eqref{eq:eqL6}, we obtain
\begin{equation}
 \| \frac{1}{M} \mathcal{A}^H(\bfeta) \| \leq  {\sqrt{2 + \delta}}\left(\frac{1}{\sqrt{M}}\| \mathcal{A}(\tilde{\X}) \| \right) \frac{\| \bfeta \|}{\| \mathcal{A}( \tilde{\X}) \|}.
\end{equation}
Invoking the upper bound from Corollary \ref{cor:Corr3} once again, along with the definition of $\mathrm{SNR}$, and Jensen's inequality, we obtain
\begin{equation}\label{eq:Sec3_5lem_prooffn}
\mathbb{E}_{\bfeta} \left[{\sqrt{2 + \delta}}\left(\frac{1}{\sqrt{M}}\| \mathcal{A}(\tilde{\X}) \| \right)\frac{\| \bfeta \|}{\| \mathcal{A}(\tilde{\X}) \|} \right] \leq \frac{2+\delta}{\sqrt{\mathrm{SNR}}} \| \tilde{\X} \|_F,
\end{equation}
and plugging the bound in \eqref{eq:Sec3_5lem_prooffn} into \eqref{eq:Sec3_5proofLem1} the proof is complete. 

\subsection{Proof of Theorem \ref{thm:Thm2}} \label{sec:Prf7}
Given the lifted domain definition of the clean gradient term in \eqref{eq:GradEq}, using the noisy measurements $\mathbf{y} = \mathcal{A}({\x \x^H}) + \bfeta$ with $\bfeta \in \mathbb{R}^M$, we define 
\begin{equation}
\nabla \tilde{f}(\z) = \left( \frac{1}{M} \mathcal{A}^H \mathcal{A}(\tilde{\Z} - \tilde{\X}) + \frac{1}{M} \mathcal{A}^H (\bfeta) \right) \z. 
\end{equation}
Having $\nabla {f}(\z)  = \frac{1}{M} \mathcal{A}^H \mathcal{A}(\tilde{\Z} - \tilde{\X})\z$ as the ideal gradient estimate and setting $\mu_{k+1} = \mu'$, we analyze the following updates:
\begin{equation}
\begin{split}
\z_{k + 1} & = (\z_{k}  - \frac{\mu'}{\| \x \|^2}  \nabla {{f}}(\z_k)) + \frac{\mu'}{\| \x \|^2} \frac{1}{M} \mathcal{A}^H (\bfeta)\z_k \\
&= \hat{\z}_{k + 1} + \frac{\mu}{\| \x \|^2} \frac{1}{M} \mathcal{A}^H (\bfeta)\z_k, 
\end{split}
\end{equation}
where $\hat{\z}_{k + 1}$ denotes the iterate obtained from the ideal update given the current estimate $\z_k$. 
We now approach the proof by induction. 
Starting from the first iteration $k = 0$, we have
\begin{equation}
\z_{1} =  \hat{\z}_{1} + \frac{\mu}{\| \x \|^2} \frac{1}{M} \mathcal{A}^H (\bfeta)\z_0, 
\end{equation}
\begin{equation}
\begin{split}
\mathrm{dist}(\z_{1}, \x) \leq \| \z_{1} - \hat{\x}_t \| &\leq \| \z_{1}  - \hat{\z}_{1} \| + \| \hat{\z}_{1} - \hat{\x} \| \\
&=  \mathrm{dist}(\hat{\z}_{1}, \x) + \| \z_{1}  - \hat{\z}_{1} \|.
\end{split}
\end{equation}
Furthermore given \eqref{eq:alphbd}, the iterates are guaranteed to stay in the $\epsilon$-neighborhood determined by the spectral initialization, which via Theorem \ref{thm:Theorem1} guarantees 
\begin{equation}
\mathrm{dist}(\hat{\z}_{1}, \x) \leq \epsilon \| \x \| (1 - \frac{2 \mu'}{\alpha'} )^{1/2},
\end{equation}
under the validity of \eqref{eq:Sec3_5_eq2}, where $\mu' / \alpha' = \mu / \alpha$. 
Repeating for the next iteration, under the validity of \eqref{eq:Sec3_5_eq2} and \eqref{eq:alphbd}, we have that
\begin{equation}
\begin{split}
\mathrm{dist}(\z_2, \x) &\leq \| \z_2 - \hat{\z}_2 \| + \mathrm{dist}(\hat{\z}_2, \x)  \\
&\leq (1 - \frac{2 \mu}{\alpha} ) \epsilon \| \x \| + \sum_{l = 1}^2 (1 - \frac{2 \mu}{\alpha} )^{\frac{l-1}{2}} \| \z_{l}  - \hat{\z}_{l} \|,
\end{split}
\end{equation}
and by induction, we obtain
\begin{equation}
\mathrm{dist}(\z_k, \x) \leq \epsilon (1 - \frac{2 \mu}{\alpha} )^{k/2} \| \x \| + \sum_{l = 1}^{k} (1 - \frac{2\mu}{\alpha})^{\frac{l-1}{2}} \| \z_l - \hat{\z}_l \|. 
\end{equation}
Recalling that $\| \z_l - \hat{\z}_l \| = \| \frac{\mu'}{\| \x \|^2} \frac{1}{M} \mathcal{A}^H (\bfeta)\z_{l-1} \|$, we can bound the term within the summation as follows:
\begin{equation}\label{eq:eqL7}
\| \z_l - \hat{\z}_l \| \leq  \frac{\mu'}{\| \x \|^2} \| \z_{l-1} \| \frac{(2 + \delta)}{\sqrt{\mathrm{SNR}}} \| \x \|^2 := U_l.
\end{equation}
After summing both sides in \eqref{eq:eqL7}, we approximately obtain
\begin{equation}
\sum_{l= 1}^k U_l \approx  \frac{(2 + \delta)}{\sqrt{\mathrm{SNR}}}  \sum_{l = 1}^{k} \mu' (1 - \frac{2\mu'}{\alpha'})^{\frac{l-1}{2}} \| \x \|
\end{equation}
where we have used that the norms of the iterates $\| \z_{l-1} \| \approx \| \x \|$. 
Then, applying the geometric sum formula, we obtain
\begin{equation}
\begin{split}
\sum_{l = 1}^k U_l & \approx \frac{(2 + \delta)}{\sqrt{\mathrm{SNR}}} \mu' \| \x \| \frac{1 - (1 - \frac{2\mu'}{\alpha'})^{\frac{k-1}{2}} }{1 - (1 - \frac{2\mu'}{\alpha'})^{\frac{1}{2}}} \\
&\leq  \frac{(2 + \delta)}{\sqrt{\mathrm{SNR}}} \mu' \| \x \|  \frac{1+ \sqrt{1 - \frac{2 \mu'}{\alpha'}}}{ (1- \sqrt{1 - \frac{2 \mu'}{\alpha'}}) (1+ \sqrt{1 - \frac{2 \mu'}{\alpha'}}) } \\
& \leq  \frac{(2 + \delta)}{\sqrt{\mathrm{SNR}}} \| \x \| \frac{2\mu' \alpha'}{2\mu'}
\end{split}
\end{equation}
which completes the proof. 

\section{Conclusion}
\label{sec:Conc}

This paper analyzes the exact recovery guarantees of the non-convex phase retrieval framework of Wirtinger Flow through a novel perspective in the lifted domain. 
Our approach quantifies a regime in which the concentration bound of the spectral matrix geometrically implies the validity of the regularity condition.  
As a result, we identify a sufficient condition under which the convergence to the true solution is guaranteed deterministically via Wirtinger Flow, starting from the estimate obtained from the spectral initialization.
Notably, our results address a theoretical gap that exists in phase retrieval literature, in which convergence arguments are predominantly probabilistic in nature. 
Furthermore, the deterministic convergence arguments developed in this paper rely on a less stringent restricted isometry type property than those of state-of-the art low rank matrix recovery methods. 
Although numerical simulations on specific problem domains are beyond the scope of this paper, our results culminate into a framework that is highly relevant to applications such as wave-based imaging, in which the underlying scattering phenomenon is typically a deterministic map.  
In future work, we will study the impact of regularization on our framework, and investigate improvements on our recovery guarantees via alternative initialization schemes to improve its recovery guarantees towards practical measurement models, such as coded diffraction patterns and 2D-Fourier slices. 
Further directions of our interest also include the study of amplitude-based loss functions to obtain a more inclusive framework beyond our specified loss function over intensity measurements. 
Namely, the superior sample complexity and convergence rates these methods were shown to enjoy motivates the study of our RIP-based approach for the analysis of amplitude-based loss functions or for possible analogies between the two problems.

\label{sec:refs}

\bibliographystyle{IEEEtran}
{
\bibliography{citations,ref1,ref1_v2,ref1_alt}}

\begin{thebibliography}{10}
\providecommand{\url}[1]{#1}
\csname url@samestyle\endcsname
\providecommand{\newblock}{\relax}
\providecommand{\bibinfo}[2]{#2}
\providecommand{\BIBentrySTDinterwordspacing}{\spaceskip=0pt\relax}
\providecommand{\BIBentryALTinterwordstretchfactor}{4}
\providecommand{\BIBentryALTinterwordspacing}{\spaceskip=\fontdimen2\font plus
\BIBentryALTinterwordstretchfactor\fontdimen3\font minus
  \fontdimen4\font\relax}
\providecommand{\BIBforeignlanguage}[2]{{%
\expandafter\ifx\csname l@#1\endcsname\relax
\typeout{** WARNING: IEEEtran.bst: No hyphenation pattern has been}%
\typeout{** loaded for the language `#1'. Using the pattern for}%
\typeout{** the default language instead.}%
\else
\language=\csname l@#1\endcsname
\fi
#2}}
\providecommand{\BIBdecl}{\relax}
\BIBdecl

\bibitem{candes2015phase}
E.~J. Candes, X.~Li, and M.~Soltanolkotabi, ``Phase retrieval via {Wirtinger}
  flow: Theory and algorithms,'' \emph{IEEE Trans. Inf. Theory}, vol.~61,
  no.~4, pp. 1985--2007, Apr. 2015.

\bibitem{Candes13b}
E.~J. Candes and T.~Strohmer, ``Phaselift: Exact and stable recovery from
  magnitude measurements via convex programming,'' \emph{Commun. Pure and Appl.
  Math.}, vol.~66, no.~8, pp. 1241--1274, Aug. 2013.

\bibitem{bendory2018non}
T.~Bendory, Y.~C. Eldar, and N.~Boumal, ``Non-convex phase retrieval from stft
  measurements,'' \emph{IEEE Transactions on Information Theory}, vol.~64,
  no.~1, pp. 467--484, 2018.

\bibitem{zhang2016reshaped}
H.~Zhang and Y.~Liang, ``Reshaped wirtinger flow for solving quadratic systems
  of equations,'' \emph{stat}, vol. 1050, p.~25, 2016.

\bibitem{GS1972practical}
R.~W. Gerchberg and W.~O. Saxton, ``A practical algorithm for the determination
  of the phase from image and diffraction plane pictures,'' \emph{Optik},
  vol.~35, pp. 237--246, 1972.

\bibitem{fienup1978reconstruction}
J.~R. Fienup, ``Reconstruction of an object from the modulus of its fourier
  transform,'' \emph{Optics letters}, vol.~3, no.~1, pp. 27--29, 1978.

\bibitem{bauschke2003hybrid}
H.~H. Bauschke, P.~L. Combettes, and D.~R. Luke, ``Hybrid
  projection--reflection method for phase retrieval,'' \emph{JOSA A}, vol.~20,
  no.~6, pp. 1025--1034, 2003.

\bibitem{Chai11}
A.~Chai, M.~Moscoso, and G.~Papanicolaou, ``Array imaging using intensity-only
  measurements,'' \emph{IOP Inverse Problems J.}, vol.~27, no.~1, pp. 1--16,
  Jan. 2011.

\bibitem{Candes13a}
E.~J. Candes, Y.~Eldar, T.~Strohmer, and V.~Voroninski, ``Phase retrieval via
  matrix completion,'' \emph{SIAM J. Imag. Sci.}, vol.~6, no.~1, pp. 199--225,
  2013.

\bibitem{Waldspurger2015}
\BIBentryALTinterwordspacing
I.~Waldspurger, A.~d'Aspremont, and S.~Mallat, ``Phase recovery, maxcut and
  complex semidefinite programming,'' \emph{Mathematical Programming}, vol.
  149, no.~1, pp. 47--81, 2015. [Online]. Available:
  \url{http://dx.doi.org/10.1007/s10107-013-0738-9}
\BIBentrySTDinterwordspacing

\bibitem{netrapalli2013phase}
P.~Netrapalli, P.~Jain, and S.~Sanghavi, ``Phase retrieval using alternating
  minimization,'' in \emph{Advances in Neural Information Processing Systems},
  2013, pp. 2796--2804.

\bibitem{goldstein2018phasemax}
T.~Goldstein and C.~Studer, ``Phasemax: Convex phase retrieval via basis
  pursuit,'' \emph{IEEE Transactions on Information Theory}, 2018.

\bibitem{hand2016elementary}
P.~Hand and V.~Voroninski, ``An elementary proof of convex phase retrieval in
  the natural parameter space via the linear program phasemax,'' \emph{arXiv
  preprint arXiv:1611.03935}, 2016.

\bibitem{bahmani2017flexible}
S.~Bahmani, J.~Romberg \emph{et~al.}, ``A flexible convex relaxation for phase
  retrieval,'' \emph{Electronic Journal of Statistics}, vol.~11, no.~2, pp.
  5254--5281, 2017.

\bibitem{chen2017solving}
Y.~Chen and E.~J. Candes, ``Solving random quadratic systems of equations is
  nearly as easy as solving linear systems,'' \emph{Communications on Pure and
  Applied Mathematics}, vol.~70, pp. 0822--0883, 2017.

\bibitem{chen2017}
J.~Chen, L.~Wang, X.~Zhang, and Q.~Gu, ``Robust wirtinger flow for phase
  retrieval with arbitrary corruption,'' \emph{arXiv preprint
  arXiv:1704.06256}, 2017.

\bibitem{Zhang2017a}
H.~Zhang, Y.~Chi, and Y.~Liang, ``Median-truncated nonconvex approach for phase
  retrieval with outliers,'' 2017, preprint.

\bibitem{wang2017solving}
G.~Wang, G.~Giannakis, Y.~Saad, and J.~Chen, ``Solving most systems of random
  quadratic equations,'' in \emph{Advances in Neural Information Processing
  Systems}, 2017, pp. 1867--1877.

\bibitem{wang2018phase}
G.~Wang, G.~B. Giannakis, Y.~Saad, and J.~Chen, ``Phase retrieval via
  reweighted amplitude flow,'' \emph{IEEE Transactions on Signal Processing},
  vol.~66, no.~11, pp. 2818--2833, 2018.

\bibitem{wang2018solving}
G.~Wang, G.~B. Giannakis, and Y.~C. Eldar, ``Solving systems of random
  quadratic equations via truncated amplitude flow,'' \emph{IEEE Transactions
  on Information Theory}, vol.~64, no.~2, pp. 773--794, 2018.

\bibitem{lu2017phase}
Y.~M. Lu and G.~Li, ``Phase transitions of spectral initialization for
  high-dimensional nonconvex estimation,'' \emph{arXiv preprint
  arXiv:1702.06435}, 2017.

\bibitem{mondelli2017fundamental}
M.~Mondelli and A.~Montanari, ``Fundamental limits of weak recovery with
  applications to phase retrieval,'' \emph{Foundations of Computational
  Mathematics}, pp. 1--71, 2017.

\bibitem{luo2019optimal}
W.~Luo, W.~Alghamdi, and Y.~M. Lu, ``Optimal spectral initialization for signal
  recovery with applications to phase retrieval,'' \emph{IEEE Transactions on
  Signal Processing}, vol.~67, no.~9, pp. 2347--2356, 2019.

\bibitem{ghods2018linear}
R.~Ghods, A.~S. Lan, T.~Goldstein, and C.~Studer, ``Linear spectral estimators
  and an application to phase retrieval,'' \emph{arXiv preprint
  arXiv:1806.03547}, 2018.

\bibitem{gao2017phaseless}
B.~Gao and Z.~Xu, ``Phaseless recovery using the gauss--newton method,''
  \emph{IEEE Transactions on Signal Processing}, vol.~65, no.~22, pp.
  5885--5896, 2017.

\bibitem{bhojanapalli2016dropping}
S.~Bhojanapalli, A.~Kyrillidis, and S.~Sanghavi, ``Dropping convexity for
  faster semi-definite optimization,'' in \emph{Conference on Learning Theory},
  2016, pp. 530--582.

\bibitem{bariscan2018}
B.~Yonel and B.~Yazici, ``A generalization of wirtinger flow for exact
  interferometric inversion,'' \emph{SIAM Journal on Imaging Sciences},
  vol.~12, no.~4, pp. 2119--2164, 2019.

\bibitem{sun2018geometric}
J.~Sun, Q.~Qu, and J.~Wright, ``A geometric analysis of phase retrieval,''
  \emph{Foundations of Computational Mathematics}, vol.~18, no.~5, pp.
  1131--1198, 2018.

\bibitem{ma2019implicit}
C.~Ma, K.~Wang, Y.~Chi, and Y.~Chen, ``Implicit regularization in nonconvex
  statistical estimation: Gradient descent converges linearly for phase
  retrieval, matrix completion, and blind deconvolution,'' \emph{Foundations of
  Computational Mathematics}, pp. 1--182, 2019.

\bibitem{sanghavi2017local}
S.~Sanghavi, R.~Ward, and C.~D. White, ``The local convexity of solving systems
  of quadratic equations,'' \emph{Results in Mathematics}, vol.~71, no. 3-4,
  pp. 569--608, 2017.

\bibitem{li2018rapid}
X.~Li, S.~Ling, T.~Strohmer, and K.~Wei, ``Rapid, robust, and reliable blind
  deconvolution via nonconvex optimization,'' \emph{Applied and Computational
  Harmonic Analysis}, 2018.

\bibitem{tu2015low}
S.~Tu, R.~Boczar, M.~Simchowitz, M.~Soltanolkotabi, and B.~Recht, ``Low-rank
  solutions of linear matrix equations via procrustes flow,'' \emph{arXiv
  preprint arXiv:1507.03566}, 2015.

\bibitem{zhu2018global}
Z.~Zhu, Q.~Li, G.~Tang, and M.~B. Wakin, ``Global optimality in low-rank matrix
  optimization,'' \emph{IEEE Transactions on Signal Processing}, vol.~66,
  no.~13, pp. 3614--3628, 2018.

\bibitem{bhojanapalli2016global}
S.~Bhojanapalli, B.~Neyshabur, and N.~Srebro, ``Global optimality of local
  search for low rank matrix recovery,'' in \emph{Advances in Neural
  Information Processing Systems}, 2016, pp. 3873--3881.

\bibitem{zhang2018much}
R.~Zhang, C.~Josz, S.~Sojoudi, and J.~Lavaei, ``How much restricted isometry is
  needed in nonconvex matrix recovery?'' in \emph{Advances in neural
  information processing systems}, 2018, pp. 5586--5597.

\bibitem{zheng2015convergent}
Q.~Zheng and J.~Lafferty, ``A convergent gradient descent algorithm for rank
  minimization and semidefinite programming from random linear measurements,''
  in \emph{Advances in Neural Information Processing Systems}, 2015, pp.
  109--117.

\bibitem{wang2016unified}
L.~Wang, X.~Zhang, and Q.~Gu, ``A unified computational and statistical
  framework for nonconvex low-rank matrix estimation,'' \emph{arXiv preprint
  arXiv:1610.05275}, 2016.

\bibitem{chi2018nonconvex}
Y.~Chi, Y.~M. Lu, and Y.~Chen, ``Nonconvex optimization meets low-rank matrix
  factorization: An overview,'' \emph{arXiv preprint arXiv:1809.09573}, 2018.

\bibitem{bandeira2014saving}
A.~S. Bandeira, J.~Cahill, D.~G. Mixon, and A.~A. Nelson, ``Saving phase:
  Injectivity and stability for phase retrieval,'' \emph{Applied and
  Computational Harmonic Analysis}, vol.~37, no.~1, pp. 106--125, 2014.

\bibitem{cai2010singular}
J.-F. Cai, E.~J. Candes, and Z.~Shen, ``A singular value thresholding algorithm
  for matrix completion,'' \emph{SIAM J. Optim.}, vol.~20, no.~4, pp.
  1956--1982, 2010.

\bibitem{recht2010guaranteed}
B.~Recht, M.~Fazel, and P.~A. Parrilo, ``Guaranteed minimum-rank solutions of
  linear matrix equations via nuclear norm minimization,'' \emph{SIAM review},
  vol.~52, no.~3, pp. 471--501, 2010.

\bibitem{recht2008necessary}
B.~Recht, W.~Xu, and B.~Hassibi, ``Necessary and sufficient conditions for
  success of the nuclear norm heuristic for rank minimization,'' in
  \emph{Decision and Control, 2008. CDC 2008. 47th IEEE Conference on}.\hskip
  1em plus 0.5em minus 0.4em\relax IEEE, 2008, pp. 3065--3070.

\bibitem{recht2011null}
------, ``Null space conditions and thresholds for rank minimization,''
  \emph{Mathematical programming}, vol. 127, no.~1, pp. 175--202, 2011.

\bibitem{oymak2011simplified}
S.~Oymak, K.~Mohan, M.~Fazel, and B.~Hassibi, ``A simplified approach to
  recovery conditions for low rank matrices,'' in \emph{Information Theory
  Proceedings (ISIT), 2011 IEEE International Symposium on}.\hskip 1em plus
  0.5em minus 0.4em\relax IEEE, 2011, pp. 2318--2322.

\bibitem{cai2013sharp}
T.~T. Cai, ``Sharp rip bound for sparse signal and low-rank matrix recovery,''
  \emph{Appl. Comput. Harmon. Anal}, vol.~35, pp. 74--93, 2013.

\bibitem{Demanet14}
\BIBentryALTinterwordspacing
L.~Demanet and P.~Hand, ``\BIBforeignlanguage{English}{Stable optimizationless
  recovery from phaseless linear measurements},''
  \emph{\BIBforeignlanguage{English}{Journal of Fourier Analysis and
  Applications}}, vol.~20, no.~1, pp. 199--221, 2014. [Online]. Available:
  \url{http://dx.doi.org/10.1007/s00041-013-9305-2}
\BIBentrySTDinterwordspacing

\bibitem{yonel2020exact}
B.~Yonel, I.-Y. Son, and B.~Yazici, ``Exact multistatic interferometric imaging
  via generalized wirtinger flow,'' \emph{IEEE Transactions on Computational
  Imaging}, vol.~6, pp. 711--726, 2020.

\bibitem{baraniuk2008simple}
R.~Baraniuk, M.~Davenport, R.~DeVore, and M.~Wakin, ``A simple proof of the
  restricted isometry property for random matrices,'' \emph{Constructive
  Approximation}, vol.~28, no.~3, pp. 253--263, 2008.

\bibitem{devore1993constructive}
R.~A. DeVore and G.~G. Lorentz, \emph{Constructive approximation}.\hskip 1em
  plus 0.5em minus 0.4em\relax Springer Science \& Business Media, 1993, vol.
  303.

\end{thebibliography}

\appendix
\subsection{Uniformity for the Gaussian Model}\label{sec:App}
We begin from the concentration bound which is known to hold for the complex Gaussian sampling model when $M \geq \mathcal{O}(N \log N)$ via \cite{candes2015phase}:
\begin{equation}\label{eq:e1}
\| \frac{1}{M} \mathcal{A}^H \mathcal{A}( \x \x^H ) - (\| \x \|^2 \mathbf{I} + \x \x^H ) \| \leq {\delta} \| \x \|^2,
\end{equation}
at a fixed ${\delta} > 0$ with probability $1 - 5\mathrm{e}^{-\gamma N}$ for an appropriately chosen $\gamma({\delta}) > 0$, for \emph{any} fixed $\x \in \mathbb{C}^N$ via unitary invariance on an event\footnote{See A.4.2 in \cite{candes2015phase} for details.} $E_0$ 
which holds with probability $1 - 4/N^2$.
We need to show \eqref{eq:e1} holds for all $\x \in \mathbb{C}^N$ uniformly in the event that $E_0$ holds. Note that due to the Hermitian property of $\Delta$, we established in Lemma \ref{lem:Lemma1} that the condition holds uniformly over $\x$ if $\v$ is fixed and vice-versa. 
We also know from Corollary \ref{cor:Corr3} that we can equivalently show that $|\langle
\Delta(\x \x^H), \x \x^H | \leq {\delta}$ holds uniformly over the unit sphere in $\mathbb{C}^N$ to establish uniformity of \eqref{eq:e1}, since the maxima over $\x$ was proven to be identical for the two conditions. 

Now, on the event that $E_0$ holds, \eqref{eq:Eqn1} by definition directly implies $| \langle \Delta ( \x \x^H) , \x\x^H \rangle_F | \leq {\delta} $, for any $\x \in \mathbb{C}^N$ with $\| \x \| = 1$ at the same probability when $M \geq \mathcal{O}(N \log N)$.
It then remains to extend this result over all $\x \in \mathbb{C}^N$ via an $\epsilon$-net argument.
Having $| \langle \Delta ( \x \x^H) , \x\x^H \rangle_F | = | M^{-1} \sum_{m = 1}^M | \a_m^H \x |^4  - 2 | \leq \tilde{\delta} $, with $\mathbb{E}[| \a_m^H \x |^4] = 2$,  we follow the methodology of \cite{baraniuk2008simple} for bounding the a quantity around its expectation uniformly. Assume the bound we have at hand holds with $\tilde{\delta}$ at any $\z$ on the unit sphere, such that $\gamma$ is properly set to have
\begin{equation}\label{eq:a1}
\text{Pr} \left( \vert \frac{1}{2M} \| \mathcal{A}(\z \z^H) \|^2  - 1 \vert \geq \frac{\tilde{\delta}}{2}  \right) \leq 5\mathrm{e}^{- \gamma N}.
\end{equation}
We first need to perform the union bound over a properly defined $\epsilon$-net of $\mathcal{S}_{\epsilon}$ on the unit sphere, and then control any perturbations when generalizing the result over to the whole domain. 
To determine the appropriate $\epsilon$-net, we first have to consider how \emph{tight} the perturbation must be bounded.
From \eqref{eq:a1} we trivially have
\begin{equation}
(1 - \tilde{\delta}/2) \leq  \frac{1}{\sqrt{2M}} \| \mathcal{A}(\z \z^H) \| \leq (1 + \tilde{\delta}/2), 
\end{equation}
which will hold for all $\z \in \mathcal{S}_{\epsilon}$ once the union bound is performed. 
Next we define $A$ as the smallest number such that
\begin{equation}\label{eq:a2}
 \frac{1}{\sqrt{2M}} \| \mathcal{A}(\x \x^H) \| \leq (1 + A) \quad \text{for all} \ \| \x \| = 1. 
\end{equation}
Hence, by definition of $A$ we have that
\begin{align}
\| \mathcal{A}(\x \x^H) \| &\leq \| \mathcal{A}(\z \z^H) \| + \| \mathcal{A}(\x \x^H - \z \z^H) \| \\
\frac{1}{\sqrt{2M}}\| \mathcal{A}(\x \x^H) \| &\leq 1 + \tilde{\delta}/2 + (1 + A) \|\x \x^H - \z \z^H\|_*
\end{align}
where we have used the fact that $\| \mathcal{A}(\x \x^H - \z \z^H) \| \leq |\lambda_1| \|\mathcal{A}(\v_1 \v_1^H)\| + |\lambda_2| \|\mathcal{A}(\v_2 \v_2^H)\|$, with $\x \x^H - \z \z^H = \lambda_1 \v_1 \v_1^H + \lambda_2 \v_2 \v_2^H$. 
Since by definition $A$ is the smallest value for which \eqref{eq:a2} holds, we have that
\begin{equation}
A \leq \tilde{\delta}/2 + (1 + A)(|\lambda_1| + |\lambda_2|).
\end{equation}
Setting $\sigma = |\lambda_1| + |\lambda_2|$, we have
\begin{equation}
A \leq  \frac{ \tilde{\delta}/2 + \sigma }{1 - \sigma},
\end{equation}
where we wish to have an $\epsilon$ such that $A \leq \tilde{\delta}$, which indicates
\begin{equation}
 \frac{ \tilde{\delta}/2 + \sigma }{1 - \eta} \leq \tilde{\delta} \ \rightarrow \ \sigma \leq \frac{\tilde{\delta}}{2(1+\tilde{\delta})}
\end{equation}
which is satisfied for $\sigma \leq \tilde{\delta}/4$ by definition since $\tilde{\delta} < 1$. 
Now we know from Lemma \ref{lem:Lemma3} that for any $\x$ with $\| \x \| = 1$, and any $\z$ where $\| \x - \z \| \leq \epsilon \| \x \|$ we have that
\begin{equation}
\sigma = \|\x \x^H - \z \z^H\|_* \leq (2 + \epsilon) \| \x - \z \| \| \x \| \leq (2 + \epsilon) \epsilon.
\end{equation}
As a result, we can determine an $\epsilon$-net on the unit sphere in the signal domain, such that a perturbation is controlled in the lifted domain. It essentially requires that $\tilde{\delta}/4 = 2 \epsilon + \epsilon^2$. Since $\tilde{\delta}/4  + 1 = (1+\epsilon)^2$, with an $\epsilon$-net $\mathcal{S}_{\epsilon}$ with
\begin{equation}
\epsilon = \sqrt{\frac{\tilde{\delta}}{4} + 1} - 1 \approx \frac{\tilde{\delta}}{8}
\end{equation}
we obtain the desired result of $A \leq \tilde{\delta}$. The lower inequality then follows as $1 - \tilde{\delta}/2 - (1 + \tilde{\delta})\tilde{\delta}/4 \geq 1 - \tilde{\delta}$ as noted in \cite{baraniuk2008simple}. 
Thereby, generalizing a union bound over the set $\mathcal{S}_{\epsilon}$ with cardinality $k$, for all $\z \in  \mathcal{S}_{\epsilon}$ we have that
\begin{equation}\label{eq:a3}
\text{Pr} \left( \vert \frac{1}{2M} \| \mathcal{A}(\z \z^H) \|^2  - 1 \vert \leq \frac{\tilde{\delta}}{2}  \right) \leq 1 - 5(24/\tilde{\delta})^k \mathrm{e}^{- \gamma N},
\end{equation}
where we used the covering number of $\mathcal{S}_{\epsilon}$ from \cite{baraniuk2008simple} via reference \cite{devore1993constructive}.
Hence the concentration bound in \eqref{eq:e1} holds uniformly over all $\x$ with $\delta = 2\tilde{\delta}$ with probability $1 - 5(24/\tilde{\delta})^k \mathrm{e}^{- \gamma N} - 4N^{-2}$ when $M = \mathcal{O}(N \log N)$, and the proof is complete. 


\end{document}